\documentclass[letterpaper,12pt,notitlepage]{article}

\usepackage{amssymb}
\usepackage{amsfonts}
\usepackage{amsmath}
\usepackage{bbm}
\usepackage{ bbold }

\usepackage{setspace} 
\usepackage[margin=1.1in]{geometry}
\usepackage[utf8]{inputenc}
\usepackage{enumerate}
\usepackage{amsmath}
\usepackage{amsthm}
\usepackage{MnSymbol,wasysym}      

\usepackage{threeparttable}
\usepackage{graphicx} 
\usepackage[svgnames]{xcolor}

\usepackage{amsfonts}

\usepackage{verbatim}
\providecommand{\R}{\mathbb{R}}
\providecommand{\R}{\mathbb{R}}
\providecommand{\E}{\mathbb{E}}

\usepackage[round]{natbib}
\usepackage[hidelinks]{hyperref}
\hypersetup{
	colorlinks=false,
	linkcolor=blue,
	filecolor=magenta,      
	urlcolor=blue,
	citecolor=blue,
	pdftitle={Sharelatex Example},
	bookmarks=true
}

\begin{document}
	
	\title{\textbf{When to Request Evidence?}\thanks{ We have received excellent research assistance from Isaac Lara. A. Espitia gratefully acknowledges support from the Deutsche Forschungsgemeinschaft (German Research Foundation) through grant CRC TR 224 (Project B02). E. Mu\~noz-Rodriguez is grateful for funding from the Ford Foundation through a PREA program grant at El Colegio de Mexico, A.C. We thank Sarah Auster, Francesc Dilmé, Rosina Rodríguez-Olivera, James Schummer, as well as the participants at various seminars, for their insightful comments and suggestions. The views expressed herein are those of the authors and do not necessarily reflect those of the German Research Foundation or the Ford Foundation.}  }
	\author{Andrés Espitia,\thanks{Kellogg School of Management, Northwestern University,  \href{mailto:andres.espitia@kellogg.northwestern.edu}{andres.espitia@kellogg.northwestern.edu}.}   Edwin Mu\~{n}oz-Rodr\'{\i}guez\thanks{Center for Economic Studies, El Colegio de M\'exico A.C.,
			\href{mailto:eamunoz@colmex.mx}{eamunoz@colmex.mx}.}}
	\date{}
	\maketitle
	\newtheorem{definition}{Definition}
	\newtheorem{assumption}{Assumption}
	\newtheorem{lemma}{Lemma}
	\newtheorem{claim}{Claim}
	\newtheorem{corollary}{Corollary}
	\newtheorem{proposition}{Proposition}
	\newtheorem{theorem}{Theorem}
	\newtheorem{remark}{Remark}
	\newtheorem{result}{Result}
	\newtheorem{example}{Example}
	
	\begin{abstract}
		\singlespacing
	Appropriate decisions depend on information gathered beforehand, yet such information is often obtained through intermediaries with biased preferences. Motivated by settings such as testing and recertification in organ transplantation, we study the problem faced by a decision-maker who can only access costly information through an agent with misaligned preferences. In a dynamic framework with exogenous decision timing, we ask how requests for verifiable information (evidence) should be scheduled and their implications for the quality of attained choices. When the agent's incentives are ignored, evidence requests do not condition on previously reported information. However, such policies may be susceptible to strategic manipulation by the agent. We show that, in these cases, optimal requests should be biased: additional evidence is more likely to be sought when previous reports favor the agent’s preferred outcome. 
	\end{abstract}
	
	
	
	\newpage

	\section{Introduction}
	
	\label{sec:introduction}
	
	Decision-makers in many settings rely on information generated by others to make high-stakes choices. When the underlying conditions evolve over time,  agents are typically required to update information at regular intervals to ensure that decisions are based on contemporaneous data.
	For instance, in the United States, the Securities and Exchange Commission (SEC) requires publicly traded companies to disclose quarterly (10-Q) and annual (10-K) earnings \citep{bushee2005economic}. Similarly, organizations often rely on periodic performance evaluations to make decisions. Finally, in the United States, transplant centers are required to regularly update health information for patients awaiting liver, heart, or lung transplantation so that authorities allocate exogenously arriving organs among many patients using up-to-date patient data. The last example is the primary motivation for the analysis that follows and it is further discussed in \autoref{sec:facts}.

	If acquiring information were costless, the decision-maker would likely request updates as frequently as possible. However, gathering information tends to be costly for both parties.\footnote{In the case of organ transplantation, updating patient health status requires them to travel to testing facilities, the tests themselves are costly, and their results must be captured into the centralized system, which creates administrative burden on transplant centers \citep{roberts2003cost}. For firms, producing reports that meet the regulators' standards is costly, and processing them is costly for the regulators themselves.}
	Therefore, planners designing schedules for updating information face a trade-off between the quality of information gathered and its costs, which go beyond the direct cost of acquiring information, and include potential distortions in the decisions themselves.

	Update schedules might also include voluntary disclosures. In practice, transplantation authorities allow transplant centers to update patient health status at any time between scheduled updates, and the SEC permits companies to make voluntary disclosures between mandatory filings. This flexibility can help bridge information gaps. However, whenever agents' objectives diverge from those of the decision-maker, discretion over whether and when to disclose information creates the scope for strategic timing of information transmission. A well-established empirical literature documents that firms strategically choose the timing of their information disclosures \citep{healy2001information, kothari2009managers, ge2011acquirers}.
	A similar pattern appears in deceased-donor liver transplantation, where priority-increasing clinical changes are reported promptly and priority-decreasing ones later (see \autoref{sec:facts}).

	We study how a planner should design an information-update schedule and an allocation policy when an agent may strategically time both the acquisition and the disclosure of verifiable information. To this end, we consider a two-period setting in which a principal commits to a deterministic mechanism without transfers. The mechanism consists of a recommendation about when the agent should acquire evidence (the \textit{testing policy}) and a rule mapping the evidence the agent chooses to disclose into a final decision at an exogenous, deterministic decision time (the \textit{assignment policy}). The agent has state-independent preferences over a binary outcome and, in each period, decides whether to acquire and disclose evidence that fully reveals the state in that period. Evidence may also arrive exogenously. For example, in transplantation, a patient may experience a health event that generates unsolicited clinical information as a by-product. 

	Our analysis proceeds in three steps. First, as a \textit{baseline}, we characterize the mechanism that would be chosen by the principal in the absence of agency frictions.
	This mechanism is \textit{result-independent} in that its testing recommendations do not depend on the history of reported test results, just on whether or not evidence was provided when requested. Likewise, the baseline (\textit{efficient}) assignment policy only conditions on payoff-relevant information (e.g., on the result of the last test submitted).

	Second, we analyze the agent’s incentives under the baseline mechanism. 
	When testing is not recommended and testing costs are low, the agent may acquire evidence nonetheless and disclose only favorable results. Given that the principal cannot distinguish a lack of evidence from withheld evidence in periods when testing was not requested, even small gaps in how assignments depend on reported results can induce such deviations, leading to \textit{ex-post} inefficient assignments and overtesting.

	Third, we characterize the optimal mechanism, taking into account the agent’s incentives to acquire and withhold information, and how these incentives are shaped by the assignment policy. For the region of the parameter space in which a deviation from the baseline testing recommendations is profitable for the agent, we show that the scope the principal has to condition the assignments on the report in histories after which testing is not requested depends critically on the agent’s cost of acquiring evidence.\footnote{When the agent faces no costs of evidence acquisition, this scope completely closes: requesting no test in a given period requires not conditioning the assignment on the report submitted that period.}
	Therefore, optimal testing policies tend to display bias and become \textit{result-dependent}. In particular, it requests evidence acquisition when the previous information supports the agent's preferred choice.
	
Moreover, when testing costs are low for both parties, the baseline testing and assignment policies are compatible with the agent’s incentives to acquire and disclose information.\footnote{ We also provide a tight lower bound in terms of the amount of mandatory information acquisition needed to attain full disclosure while preserving the baseline assignment policy.}  
Finally, when information acquisition is ``expensive" for both players, the optimal testing policy coincides with the baseline by never recommending acquiring evidence 
but also requires distortions in the assignment policy relative to the (\textit{efficient}) baseline one.

	All things considered, our main contribution is to establish that \textit{result-independent policies} allow for opportunistic behavior from the agent (in the form of acquiring evidence when it was not requested and censoring the reports according to his preferences). Thus, optimal policies must take such behavior into account. This leads to the optimality of \textit{result-dependent} policies that condition testing recommendations and/or assignments on previous test results (even when such information is payoff irrelevant). 
	
	Furthermore, we present an extension of the model with a continuum of agents. This analysis provides a foundation for our focus on the single-agent case.
		
	\subsection{Related literature} 
	
	This paper relates to four  (broadly defined) branches of the literature: dynamic information acquisition, dynamic allocation without money, dynamic moral hazard, and (more generally) mechanism design problems without money and verifiable information.
	
	Regarding \textit{centralized} dynamic information acquisition,  \cite{zhong2022optimal} studies the problem with endogenous timing of choice (involving an optimal stopping time). Complementarily, \cite{georgiadis2024information} focuses on an exogenous timing of choice. They analyze very general information acquisition technologies. In contrast, we focus on \textit{decentralized} information acquisition with an agent with different preferences than the decision-maker's. Given that the decision-maker in our framework relies on a second player for information acquisition, we assume that information is verifiable. Furthermore, we simplify the information acquisition technology by focusing on the simple case of either learning the current state with certainty or learning nothing.  
	
	There are several instances of decentralized information acquisition with biased agents. In a continuous time setting, \cite{bardhi2024attributes} studies an agent who decides which dimensions (or attributes) of a complex and invariant state to learn about. Information acquisition costs are not considered, focusing on how the agent optimally chooses which attributes to learn about. In contrast, \cite{ball2023should} consider how the agent's choices affect the arrival rate of information (in the form of \textit{breakthrough} or a \textit{breakdown}), which is directly observed by the decision-maker. Therefore, information elicitation is not an issue in their setting. Moreover, the decision-maker has full control of when to end the interaction with the agent.

	Closer to our approach are \cite{li2023dynamics} and \cite{liimplementing}. On one hand, similar to the current setting, \cite{li2023dynamics} consider an agent acquiring information over time about a changing state, biased towards a particular choice by the decision-maker, and where transfers are unavailable. Information is assumed to be verifiable, provided the decision-maker pays a cost. Crucially, the decision-maker faces an optimal stopping problem, which is absent in our setting. On the other hand, \cite{liimplementing} consider non-verifiable information and allows the decision-maker to commit to (bounded) state-dependent transfers.

	Moreover, dynamic allocation problems without transfers have been analyzed in several instances. We proceed to discuss a strict subset of them. First, \cite{guo2020dynamic} and  \cite{chen2022dynamic} consider environments with evolving states and an agent with biased preferences.  However, the agent is always privately informed about the ``state" (which can be interpreted as his private type). Hence, the information acquisition problem is inactive.   \cite{chen2023ask} allow the agent to conceal (but not misrepresent) his private information (which is also free in their setting). They compare two reporting protocols: \textit{frequent updating} (which asks for a report in each period) and \textit{infrequent updating} (which only requires a report at the end of the learning process). 
	
	In contrast,  \cite{ashlagi2024optimal} focus on the case where the decision-maker can design the information disclosed to a mass of agents. Moreover, the agent's private information is non-verifiable, requiring incentives for the agent's truthful information revelation. Finally, \cite{guo2016dynamic} studies delegated experimentation to an agent who decides how to allocate resources between a safe and a risky project. Similar to our approach, transfers are ruled out. However, decisions are reversible and adjustable. The main result establishes the optimality of a cutoff rule (in the belief space).

	Moral hazard in a dynamic setting is studied, among others, by \cite{sannikov2008continuous}, \cite{biais2010large}, and \cite{cvitanic2017moral} under the assumption of full commitment on the side of the decision-maker. The case without commitment is studied by \cite{horner2016dynamic}, while \cite{ma1991adverse} studies frictions related to renegotiation. We depart from their analysis by ruling out contingent payments and assuming that the decision-maker does not directly observe the output, capturing cases in which progress on a project is privately observed by the agent.

	Models of mechanism design without money and verifiable information include  \cite{deneckere2008mechanism}, \cite{ben2014optimal},   \cite{rosar2017test}, \cite{libgober2022false},  and \cite{ben2023sequential}. However, the underlying source of uncertainty does not evolve in these models. We contribute to this literature by considering the question of the optimal \textit{timing} of evidence acquisition and disclosure.

	Turning to our leading application, a growing literature in economics studies informational frictions in deceased-donor organ allocation. Closest to our work are theoretical papers that analyze the strategic misrepresentation of patients’ medical needs to gain priority for organs, either in environments where misrepresentation is costly \citep{munoz2024ethical, munoz2024rationing, perez2025falsification} or in settings of cheap talk \citep{ashlagi2024optimal}. On the empirical side, beyond \citet{gordon2021chapter3}, discussed in detail in \autoref{sec:facts}, a number of papers study how uncertainty about future organ quality shapes transplant acceptance decisions, and analyze the resulting equilibrium implications for assignments, treatment choices, and learning \citep{agarwal2025choices, agarwal2020equilibrium, agarwal2018dynamic, doval2024social, howard2002transplant, zhang2010sound}. \citet{snyder2010gaming} is the first paper in economics to provide empirical evidence of gaming of treatments in liver transplantation. Other empirical studies of deceased-donor organ allocation include \citet{dickert2019allocating}, \citet{genie2020role}, and \citet{callison2025externalities}.

Finally, the unknown binary outcome in our model can be interpreted as indicating whether the agent has high or low priority for assignment. This connects our paper with the literature on designing priorities through waiting lists, which includes \cite{su2004patient, bloch2017dynamic, schummer2021influencing, leshno2022dynamic, arnosti2020design, thakral2016publichousing, murraanton2024publichousing} and \cite{ margaria2025queueing}.

	\section{Motivating facts}
	\label{sec:facts}
	
In deceased-donor organ transplantation, the timing of allocation decision is exogenous: procured organs arrive stochastically, outside the control of transplantation authorities or transplant candidates. 
Moreover, these organs must be allocated quickly to prevent their deterioration \citep{cesaretti2024cold, debout2015each, kettlewell2023minutesmatter}. 

Because up-to-date information on transplant candidates’ health status is critical for ensuring an efficient and equitable allocation of deceased-donor organs, the U.S. transplantation authorities have established formal update schedules whenever the allocation depends on time-varying covariates. In particular, in the U.S. liver transplantation system, when an organ from a deceased donor is procured, transplant candidates are prioritized through rankings which are a function of how severely ill patients are \citep{polyak2021evolution}. Since 2002, for most transplant candidates the severity of illness has been determined by the Model for End-Stage Liver Disease (MELD) score, which depends on a set of laboratory results. A higher MELD score indicates a sicker patient.\footnote{Currently, there are three additional factors that influence the patient priority for a liver: geography, blood compatibility and time spent waiting for a liver. See \cite{us2020organ} for a review of the current liver transplant allocation policy in the United States.} 

MELD scores within a patient change over time either downward or upward, reflecting temporary improvements or deteriorations in a candidate’s clinical condition \citep{gordon2021chapter3}. Therefore, allocation policies require transplant centers to periodically recertify each patient’s score. 
The corresponding recertification schedule depends on severity: for the most critically ill group of patients, MELD must be updated every seven days, whereas for less severe groups the required intervals are one month, three months, or one year.\footnote{OPTN Allocation Policies, Section 9.2. Available at \url{https://optn.transplant.hrsa.gov/media/eavh5bf3/optn_policies.pdf}.}

In addition to these mandatory recertifications, voluntary updates of MELD scores are permitted within each group’s update window. This degree of discretion at the transplant center/patient level creates the key friction that our stylized model formalizes and that has been extensively documented by \cite{gordon2021chapter3}. Using administrative records from the U.S. transplantation authority, \citet{gordon2021chapter3} document clear patterns of strategic timing behavior in MELD recertification data: changes in MELD scores that increase a patient’s priority for a deceased-donor liver are reported promptly, whereas changes that lower a patient’s priority are systematically delayed until the end of the permissible reporting window.

In the remainder of the paper, we first formalize the channel highlighted in these stylized facts with a streamlined model. We then use the model to examine how the optimal information‐update requirements should be designed when such discretion over voluntary reporting is available. As discussed in \autoref{sec:introduction}, this discretion may be desirable from the principal’s perspective, as it can help bridge information gaps. Hence, it is not immediate that the principal would find it optimal to fully eliminate it. 

	\section{Model}
	\label{model}
	
In this section, we introduce a two-player, two-period, binary-state model. Our analysis focuses on the interplay between the cost of testing, voluntary reporting, and the design of information-update schedules. To keep a parsimonious framework, we deliberately abstract from several relevant features discussed in our motivating application, such as the stochastic arrival of organs, geographical considerations, and tissue and blood compatibility, among others.

\paragraph{Preliminaries}	Two players, a principal ($P$, she) and an agent ($A$, he), interact over two periods, $t\in\left\{1,2\right\}$. At the end of the second period, $P$ must choose an action $x\in\{0,1\}$, which we will interchangeably refer to as an \textit{assignment} or an \textit{allocation}.
		
In every period $t$, there is a binary state of the world $\omega_t\in\Omega:=\left\{0,1\right\}$ which evolves stochastically with some degree of persistency.  Specifically, the transition probabilities are given by $Pr\left[\tilde{\omega}_{t+1}=\omega|\omega_{t}=\omega\right]=\rho$ with $\rho\in(1/2,1)$.\footnote{We use tildes to denote that a variable is random. For instance, $\tilde{\omega}_t$ represents the random state at $t$ and $\omega_t$ its realization.}  The common prior belief is given by $\mu_0:=Pr[\tilde{\omega}_1=1]\in(0,1)$.
	
From $P$'s perspective, the appropriate \textit{allocation} depends on the realized state of the world at $t=2$, $\omega_2$. In contrast, the states $\omega_1$ and $\omega_2$ are payoff-irrelevant for $A$. Neither $P$ nor $A$ can directly observe $\omega_t$ at any point in time.

 \paragraph{Information acquisition and transmission} While at any period $t$ the state of the world is hidden from both players, each period $A$ can take a \textit{test} which fully reveals $\omega_t$.  We denote $A$'s testing choice in $t$ by $e_t\in\{0,1\}$. Testing in $t$ generates a private cost $c\in\R_{++}$ to $A$ and $k\in\R_{++}$ to $P$. 
	 
	Additionally, if $A$ decides not to acquire information, i.e. $e_t=0$, he may still learn the state of the world with probability $\pi\in(0,1)$. 
	
Furthermore, once $A$ observes a test result, he decides whether or not to report it. Let $r_t\in\left\{\O,\omega_t\right\}$ denote the report that the agent sends to the principal, where the null report $r_t=\O$ indicates that it is not a test result. Hence, we are assuming the non-falsifiability of information.
		
\paragraph{Timing}
	In each period $t$, first, nature  selects $\omega_t$. Afterwards,  $A$ decides whether to take the test, $e_t$.  If the test is taken, the agent observes $\omega_t$ with probability 1.  If the test is not taken, the agent observes $\omega_t$ with probability $\pi$.  Next, $A$ sends a report $r_t$ to $P$ who in turn chooses $x$ if $t=2$.

\paragraph{\textit{Ex-post} payoffs}	For an assignment $x$ and testing decisions $e=(e_1,e_2)$, $P$ receives $\mathbb{1}[x=\omega_2]-k(e_1+e_2)$, while $A$ receives $x-c(e_1+e_2)$, where $\mathbb{1}[\cdot]$ denotes the indicator function. Importantly, testing costs are incurred only at the end of the interaction, which prevents the principal from making inferences about the agent’s testing decisions at any interim stage.

	\subsection{Mechanisms}
	
	The principal is assumed to have commitment power but monetary transfers are ruled out. We restrict attention to deterministic mechanisms.

	A \textit{public history} $h_t:=\left(\ldots,r_{t}\right)$ keeps track of the reports submitted by the agent at the end of that period. Let $\mathcal{H}_t:=\left\{\O,0,1\right\}^t$ denote the set of all possible histories at the end of period $t$. 
	Before the beginning of the interaction $P$ sets up a mechanism $m=\left(\sigma,\hat{x}\right)$ consisting of a \textit{testing policy} $\sigma$ and an \textit{assignment policy} $\hat{x}$.
	
	A \textit{testing policy} is a vector $\sigma=(\sigma_1,\sigma_2)$ where $\sigma_1\in\{0,1\}$ and $\sigma_2:\mathcal{H}_1\to\left\{0,1\right\}$  recommending whether to submit a test result in a given period, as a function of the public history. Thus, a testing policy $\sigma$ is a four-dimensional binary vector.
	Examples of feasible testing policies include: never require a test result, requiring a test in both periods, requiring a test only in $t=2$ regardless of the history, or requiring a test in $t=2$ if and only if a given result is reported in $t=1$.	
	
	An \textit{assignment policy} is a vector  $\hat{x}:\mathcal{H}_2\to \{0,1\}$ prescribing an action as a function of the public history at the end of $t=2$.  Examples of feasible assignment policies include: basing the allocation solely on $P$'s beliefs about the state of the world at the time of the decision, punishing $A$ with his least preferred choice whenever he does not comply with the testing recommendation, condition the assignment only on the last submitted test result, committing to a particular allocation independently of the realized history.
	We denote by $\hat{x}(r_1,r_2)$ the assignment prescribed by the policy after reports $(r_1,r_2)$. Since, there are nine possible histories at the end of $t=2$, $\hat{x}$ is a vector in $\{0,1\}^9$.

	Let $\mathcal{M}$ be the set of such mechanisms, which has, in principle, $2^4\cdot 2^9=8192$ elements. Thus, as with any finite optimization problem, a solution always exists. As an alternative to comparing $P$'s expected payoffs under each alternative, we will progressively limit the search of the optimal mechanism to a more manageable subset of $\mathcal{M}$. Whenever $A$ is indifferent between two actions, we adopt the tie-breaking rule that he follows the test recommendation and discloses the test result
	
	\section{Baseline mechanism}

	In this section, we characterize the baseline mechanism that would emerge in the absence of agency frictions. We then analyze $A$’s behavior under this benchmark, finding a region of testing costs where $A$ can strictly improve by strategically timing reports. Finally, we illustrate how the mechanism can be improved once $A$’s incentives are explicitly taken into account.

	Suppose $P$ makes all decisions and directly observes test results, fully ignoring $A$'s preferences. The problem can be formulated as 
	\begin{align*}
		\max_{(\sigma,\hat{x})\in\mathcal{M}}\ Pr_{\sigma}[\hat{x}(\tilde{r}_1,\tilde{r}_2)=\tilde{\omega}_2]-k\cdot (\sigma_1+\E_{\sigma_1}[\sigma_2(\tilde{r}_1)]),
	\end{align*}
	where 
	\[\E_{\sigma_1}[\sigma_2(\tilde{r}_1)]:=[\sigma_1+(1-\sigma_1)\pi][\mu_0 \sigma_2(1)+(1-\mu_0)\sigma_2(0)]+(1-\sigma_1)(1-\pi)\sigma_2(\O)\]
	 is the probability of testing in $t=2$, given the testing recommendation for $t=1$.	

	Let $\mu_t(h_t):=\E[\tilde{\omega}_t|h_t]$ be the belief about $\tilde{\omega}_t$ at the \textit{end} of $t$ and \[\mu_t(h_{t-1}):=\E[\tilde{\omega}_t|h_{t-1}]=\rho\mu_{t-1}(h_{t-1})+(1-\rho)(1-\mu_{t-1}(h_{t-1}))\]
	 be the belief about $\tilde{\omega}_t$ at the \textit{beginning} of $t$.  Naturally, at $t=1$, $\mu_t(h_{t-1})=\mu_0$. Additionally, \[\mu_2(\O)=\rho\mu_0+(1-\rho)(1-\mu_0).\]
	
	We proceed with the analysis by backward induction. Given beliefs $\mu_2(h_2)$, $P$ chooses $x=1$ if $\mu_2(h_2)> 1/2$, $x=0$ if $\mu_2(h_2)<1/2$, and she is indifferent between the two actions when $\mu_2(h_2)=1/2$.  
	
	\begin{definition}[Efficient assignments]
			An assignment policy is \textit{efficient} if, for all histories $h_2$ that occur with strictly positive probability, it prescribes $P$'s static optimal decision given her beliefs derived from public reports at the time of choice; that is,$
		\hat{x}^*(h_2):=\mathbb{1}[\mu_2(h_2)\geq 1/2].
		$
	\end{definition}

	In the previous definition, we emphasize that $\mu_2(h_2)$ is $P$'s belief based solely on the history of public reports. 
	In the baseline case, this coincides with the belief formed from all information produced in the economy. In later sections, however, $A$'s strategic behavior may generate a gap between public reports and privately gathered information. When this gap is not accommodated, decisions can be inefficient from an \textit{ex-post} perspective, despite appearing efficient when only public reports are considered.

	The next proposition states that the baseline mechanism pairs efficient assignments with an intuitive testing policy. Never test in the first period, and testing in the second period depending
	 on $P$'s \textit{effective cost} of testing, $\kappa:=k/(1-\pi)$. When this cost is small, always testing in the second period (when the information is more valuable) is preferred, while if the cost is large enough, never testing is preferred. For intermediate costs, however, $P$ prefers to test in $t=2$ only if she did not get any informative signal in $t=1$.

	\begin{proposition}[Baseline mechanism]
		\label{prop:baseline}
		The baseline mechanism pairs the efficient assignments $\hat{x}^*$ with the following testing policy:	$\sigma_1^*=0$ and 
		\begin{itemize}
			\item If $\kappa\in(0, 1-\rho]$:   $\sigma_2^*(r_1)=1$ for all $r_1$. 
			\item  If $\kappa\in(1-\rho,\min\{\mu_2(\O),1-\mu_2(\O)\}]$:	$\sigma_2^*(r_1)=\mathbb{1}[r_1=\O]$. 
			
			\item If $\kappa>\min\{\mu_2(\O),1-\mu_2(\O)\}$: $\sigma_2^*(r_1)=0$ for all $r_1$.  
		\end{itemize}
		Moreover, if $\kappa= 0$, then always testing is optimal.
	\end{proposition}
	
	Note that whether or not a test is recommended in $t=2$ does not depend on $\omega_1$. The recommendation can depend on whether or not $r_1=\O$, but it is independent of the result itself. Further note that the efficient assignments do not condition on any payoff irrelevant information (i.e., if $r_2\neq\O$, then $r_1$ plays no role). Therefore, we say that this baseline mechanism is \textit{result-independent}.

	\subsection{Punishments}

		Whenever $A$ deviates from a recommendation to test after a given history, $P$ can punish $A$ with an allocation of zero. Specifically, for the baseline mechanism, $P$ can set $\hat{x}_2(\O,\O)=0$. 
	
	In order to formalize the previous idea, consider any mechanism $m=(\sigma,\hat{x})\in\mathcal{M}$. Let $e^m_t(h_t)\in\{0,1\}$ denote the agent's optimal testing decision in $t$ when the history at the beginning of the period is $h_t$. Similarly, let $r^m_t(h_t,e_t,\omega_t)\in\{\omega_t,\O\}$ denote the agent's optimal disclosure decision when the history is given by $h_t$, the agent testing decision was $e_t$, and the test result is $\omega_t$. If no test result is observed, $A$ has no disclosure decision to make. Naturally, for $t=1$ we can write $e^m_1$ and $r^m_1(e_1,\omega_1)$ since there is no history at that point. 
	
	Hereafter, we assume that any detected deviation from the recommendation is punished with $A$'s least preferred assignment.  Specifically, if a test is requested at some history and no result is reported, then the principal assigns the agent his worst outcome at that history and sets all future test recommendations to zero. Moreover, because $\sigma_1=1$ is never optimal, prescribing punishments for deviating from $\sigma_1=1$ are unnecessary.

	\begin{assumption}[Forcing mechanisms]
		\label{assump:forcing}
		For any $r_1$ with $\sigma_2(r_1)=1$, $r^{m}_2(r_1,e^m_2(r_1),\omega_2)=\O$ implies  $\hat{x}(r_1,\O)=0$.
	\end{assumption}

		We will establish in \autoref{subsec:IC} that this assumption is without loss of generality. 
	\subsection{Agent's behavior under the baseline mechanism}
		
We now turn to $A$’s optimal behavior under the baseline mechanism described in \autoref{prop:baseline}. 
$A$’s compliance with the baseline testing policy is naturally determined by his own testing costs $c$. For clarity of exposition, we focus on cases in which both parties have intermediate testing costs.\footnote{Other possible combinations of testing costs are considered in Appendix C. The approach developed in \autoref{sec:optimalmech} applies to any combination of parameters.} Thus, the following assumption will be maintained throughout.

\begin{assumption}(Intermediate testing costs for $P$)
	\label{assump:interKappa}
\[\kappa\in(1-\rho,\min\{\mu_2(\O),1-\mu_2(\O)\}].\]
\end{assumption}

Under \autoref{assump:interKappa} the baseline mechanism recommends testing in $t=2$ if and only if no test result was reported in $t=1$, and no testing in $t=1$. By \autoref{assump:forcing}, $P$ sets $\hat{x}_2(\O,\O)=0$.

	However, under a recommendation not to test, $A$ can obtain evidence at no cost with positive probability ($\pi$) while complying with the recommendation. This information is not requested, yet it remains informative for $P$. Thus, $P$ faces a trade-off between exercising control and facilitating information acquisition, which opens the door for $A$ to engage in discretionary testing when not required and to exploit the opportunity to conceal unfavorable outcomes.\footnote{If we close this channel ($\pi=0$), $P$ can always get $A$ to follow the baseline mechanism, provided $\gamma\leq \mu_2(\O)$, fully disclosing all acquired information. See \autoref{lem:imppi0} in the Appendix.} Whether $A$ engages in discretionary testing or not depends on $A$'s testing costs, as the following lemma states.

	\begin{lemma}
		\label{lem:testObedt2}
		 
		 Let $\gamma:=c/(1-\pi)$ be $A$'s effective costs of testing. 
		 $A$ follows $\sigma_2^*(1)=\sigma_2^*(0)=0$ if and only if $\gamma>1-\rho$. Moreover:
		 
		\begin{itemize}
			\item If $\gamma\leq 1-\rho$, $A$ follows $\sigma_2^*(\O)=1$.
			\item If $\gamma \in (1-\rho,\mu_2(\O)]$, $A$ follows $\sigma_2^*(\O)=1$ if he disclosed all available information at $t=1$. 
			\item If $\gamma>\mu_2(\O)$, $A$ never follows $\sigma_2^*(\O)=1$. 
		\end{itemize}
	\end{lemma}
	
		Therefore, either $\gamma>1-\rho$ and $A$ follows $\sigma_2^*(0)=\sigma_2^*(1)=0$ or $\gamma\leq1-\rho$ and $A$ follows $\sigma_2^*(\O)=1$ regardless of the history.  
		Thus, $A$ does not follow the baseline testing recommendation in every history in $t=2$. However, when $\gamma\in(1-\rho,\mu_2(\O)]$,  $A$ follows $\sigma_2^*(\O)=1$ \textit{as long as} he disclosed the information available in $t=1$.  
		
	Now we take into account $A$'s behavior in $t=1$. In the following definition, we present a strategy for $A$ that deviates from the recommendations given by the baseline mechanism.

	\begin{definition}[Deviation strategy]
		\label{def:deviation}
	$e^{m^*}_1=1$, $r_1^{m^*}(e_1=1,\omega_1)=1$ if and only if $\omega_1=1$ otherwise it equals $\O$, $e^{m^*}_2(\cdot)=0$, and $r_2^{m^*}(\O,e_2=0,\omega_2)=1$ if and only if $\omega_2=1$ otherwise it equals $\O$, where $m^*$ is the baseline mechanism modified to be forcing.
	\end{definition}

	The following proposition establishes that when $A$’s effective testing costs are intermediate (formally defined by the following assumption), then the deviation previously described is profitable for the agent if and only if $A$'s effective costs fall in the lower end of that intermediate interval.\footnote{In \autoref{subsec:lowgamma} we present an analogous result  (\autoref{proposition:profdev}) for the case in which $\gamma\leq 1-\rho$.}

	\begin{assumption}[Intermediate testing costs for $A$]
		\label{assump:interGamma}
	\[\gamma\in(1-\rho,\mu_2(\O)].\]
	\end{assumption}

		\begin{proposition}[Profitable deviation from baseline]
		\label{prop:gaming1}
			Consider the case in which \autoref{assump:forcing}, \autoref{assump:interKappa} and  \autoref{assump:interGamma} hold. The deviation in \autoref{def:deviation}  is strictly profitable for $A$ if and only if
		\begin{align*}
			\label{cond1}
			\gamma< \bar{\gamma}:=\frac{(1-\rho)[\mu_0-(1-\mu_0)(1-\pi)]}{\pi}.
		\end{align*}
	\end{proposition}

Note that $\bar{\gamma}>1-\rho$ is equivalent to $\mu_0>1/(2-\pi)\in(1/2,1)$. Thus, a necessary condition for this deviation to be profitable is that prior belief is sufficiently high. Moreover, $\bar{\gamma}<\mu_2(\O)$ always holds.

\subsection{Improvements over the baseline mechanism}

The deviation described in \autoref{def:deviation} may generate a gap between public reports and privately gathered information. When this gap is not accommodated, allocation decisions can be inefficient \textit{ex-post}, despite appearing efficient when only public reports are considered. Moreover, relative to the baseline, the deviation increases the amount of testing and shifts testing away from the appropriate timing.

A particular \textit{result-dependent} testing policy, which conditions testing in $t=2$ not only on whether there was or not a report in $t=1$, as in the baseline, but on whether the news in the first period are good or bad from $A$'s perspective, restores $P$'s control and eliminates the agent's informational advantage.

\begin{lemma}[Inducing full disclosure under efficient assignments]
	\label{lem:implemEffAssign}
	Let 
	\[\sigma^{RD1}:=(\sigma_1=0,\sigma_2(0)=0,\sigma_2(1)=1,\sigma_2(\O)=1).\] If $\gamma\in(1-\rho,\mu_2(\O)]$, then  $\sigma^{RD1}$ induces $A$ to fully disclose all the information he acquires when $P$ commits to efficient assignments. 
\end{lemma}

Compare $\sigma^{RD1}$ with the baseline testing policy $\sigma^*$. When \autoref{assump:interKappa} holds, this policy tests more than the baseline ($\sigma_2(1)=1$ vs. $\sigma_2^*(1)=0$).
It turns out to be the case that $\sigma^{RD1}$ provides the minimum amount of testing required, and thus the cheapest way of inducing the agent to fully disclose all the information he acquires.

\begin{lemma}[Minimum cost of full disclosure under efficient assignments]
	\label{lem:lowerbound}
	Assume $\gamma\in(1-\rho,\mu_2(\O))$. If a mechanism induces full disclosure under efficient assignments, then $\sigma\geq \sigma^{RD1}$.\footnote{By $\sigma'\geq \sigma$ we mean that each entry of $\sigma'$ (seen as a 4-vector) is weakly above the corresponding entry of $\sigma$.}
\end{lemma}

The next step is to ask whether altering the baseline efficient assignments, coupled with a reduction in the amount of testing, can improve $P$'s expected payoffs. In the next section, we conduct that analysis.

	\section{Optimal mechanisms}
	\label{sec:optimalmech}
	
	In this section, we develop a systematic approach to characterizing the optimal mechanism over arbitrary regions of the parameter space. For expositional convenience, some results are stated for the case in which the deviation defined in \autoref{def:deviation} is profitable for the agent, though the method applies more generally. In \autoref{appen:additional}, we analyze optimal mechanisms under alternative parameter configurations.	
		
	\subsection{Incentive compatibility}
	\label{subsec:IC}	
	
	Before addressing the full problem faced by $P$ in designing an optimal mechanism, we  establish that it suffices to restrict attention to mechanisms satisfying the properties specified in the following definition.

	\begin{definition}[Incentive compatible mechanisms]
		A mechanism $m\in\mathcal{M}$ is said to be\textit{ incentive compatible (IC)} if it satisfies 
		\begin{enumerate}[i.]
			\item Obedience : $e^{m}_t(h_t)=\sigma_t(h_t)$ for all $t$ and all histories $h_t$ that occur with positive probability under $m$.
			\item  Full disclosure: $r^{m}_t(h_t,e_t,\omega_t)=\omega_t$ for all $(t,e_t,\omega_t)$ and all $h_t$ that occur with positive probability under $m$.
			\item \autoref{assump:forcing}.
		\end{enumerate}

	\end{definition}
	
	Hereafter, we will denote the obedience and full disclosure constraints by $[OB,\sigma_t(h_t)]$ and $[FD,h_{t-1},\omega_t]$, respectively.
		
	The following result, analogous to theorem 3 in \cite{deneckere2008mechanism} in the current setting, allows us to limit our attention to $IC$ mechanisms. 

	\begin{proposition}
		\label{prop:revelation}
		For any mechanism $m\in\mathcal{M}$, there exists an IC mechanism $m'=(\sigma',\hat{x}')\in \mathcal{M}$ 
		that generates the same outcomes as $m$.
		
	\end{proposition}	
	
	We adopt the tie-breaking rule that \textit{when indifferent}, $A$ follows the test recommendation and discloses the test result.
		
	The following example provides the $IC$ mechanism that generates the same outcomes as the baseline.
	\begin{example}
		As shown in \autoref{prop:gaming1}, when $\gamma\in (1-\rho,\bar{\gamma})$, the baseline mechanism generates the same outcomes as the following $IC$ mechanism: $\sigma_1=1$, $\sigma_2(r_1)=\mathbb{1}[r_1=0]$, and \[\hat{x}(r_1,r_2)=\mathbb{1}[r_1=1]+\mathbb{1}[r_1=0]\mathbb{1}[r_2=1].\]
		Note that while the baseline mechanism gives testing recommendation that are independent of previously submitted test results, the corresponding $IC$ mechanism is \textit{result-dependent}. 
	\end{example}
		
	Finally, define $A$'s expected payoff of sending report $r_1$ after observing $\omega_1$  as follows,
	\begin{align*}
	v(r_1,\mu_2):=&-c\sigma_2(r_1)+[\sigma_2(r_1)+[1-\sigma_2(r_1)]\pi][\mu_2 \hat{x}(r_1,1)+(1-\mu_2)\hat{x}(r_1,0) ]\\
	&+[1-\sigma_2(r_1)](1-\pi)\hat{x}(r_1,\O),
	\end{align*}	
	where $\mu_2:=\mathbb{1}[\omega_1=1]\rho+\mathbb{1}[\omega_1=0](1-\rho)$. Whenever we are considering $IC$ mechanisms, we will drop the dependence of $v$ on $\mu_2$, since it is clear that $\mu_2$ follows from $r_1$. However, when considering deviations we will point out which beliefs $A$ is using to compute his expected payoff.

	We next leverage the properties of $IC$ mechanisms to substantially simplify the formulation of $P$'s problem.

	On one hand, $A$ optimally discloses a test result $\omega_2$ in $t=2$ if and only if it leads to a (weakly) higher allocation than $r_2=\O$, regardless of any previous report $r_1$. On the other hand, in $t=1$, $A$ also takes into account the effect that the report $r_1$ has on future test recommendations and allocations. 
	
	Formally, after report $r_1$ (leading to testing recommendation $\sigma_2$) and observing $\omega_2$, $A$ optimally discloses this result if and only if 
	\begin{align*}
		&[FD,r_1,\omega_2]&\quad \hat{x}^{\sigma_2}(r_1,\omega_2)\geq \hat{x}^{\sigma_2}(r_1,\O)\quad\forall r_1\text{ and } \omega_2\in\{0,1\}.
	\end{align*}
	Moving to the first period, upon observing $\omega_1$, $A$ optimally discloses this result if and only if 
	\begin{align*}
		&[FD,\omega_1]&\quad v(\omega_1,\mu_2(\omega_1))\geq v(\O,\mu_2(\omega_1))\quad\text{for } \omega_1\in\{0,1\}.
	\end{align*}
	
	The next lemma follows from the mechanism being forcing.

	\begin{lemma}
		\label{lem:FDfree}
		If $[OB,\sigma_t(h_t)=1]$ holds then $[FD,h_t,\omega_t]$  is implied by the mechanism being forcing.
	\end{lemma}	
	
	Thus, we only need to consider the full disclosure constraints in $t$ when a test was not requested in that period.
	
	Given full disclosure, the belief at the beginning of the second period ($\mu_2$) is public. Moreover, $\mu_2(r_1)=\mathbb{1}[r_1=1]\rho+\mathbb{1}[r_1=0](1-\rho)+\mathbb{1}[r_1=\O]\mu_2(\O)$. In order to simplify notation, we do not explicitly indicate this dependence.
	
	For obedience, two different constraints can operate in each period: one for each possible testing recommendation. The choice of testing policy determines which constraint is active. Full disclosure, together with the description of the mechanism, fully pins down $A$'s expected payoffs for each report.
	
	Consider a testing policy with $\sigma_2(r_1)=1$ for some $r_1\in\{0,1,\O\}$. On one hand, if $A$ follows this recommendation he expects to get $-c+\E_{\mu_2}[\hat{x}(r_1,\tilde{r}_2)]$. On the other hand, by not testing he expects $\E_{\mu_2}[\hat{x}(r_1,\tilde{r}_2)]$ with probability $\pi$ and 0 otherwise.  Note that when limiting attention to forcing mechanisms, we know that $\sigma_2(r_1)=1$ implies $\hat{x}(r_1,\O)=0$. Therefore, $A$ obeys  $\sigma_2(r_1)=1$ if and only if
	\begin{align*}
		&[OB,\sigma_2(r_1)=1]&\quad\gamma\leq \E_{\mu_2}[\hat{x}(r_1,\tilde{r}_2)]\quad\forall r_1\text{ with }\sigma_2(r_1)=1.
	\end{align*}
	Now, if $\sigma_2(r_1)=0$ for some $r_1$, we have that by following this recommendation $A$ expects to get $\E_{\mu_2}[\hat{x}(r_1,\tilde{r}_2)]$ with probability $\pi$ and $\hat{x}(r_1,\O)$ with probability $1-\pi$. In contrast, by testing $A$ expects a payoff of $-c+ \E_{\mu_2}[\hat{x}(r_1,\tilde{r}_2)]$. As a result, $A$ obeys this recommendation if and only if 
	\begin{align*}
		&[OB,\sigma_2(r_1)=0]&\quad\gamma\geq \E_{\mu_2}[\hat{x}(r_1,\tilde{r}_2)]-\hat{x}(r_1,\O)\quad\forall r_1\text{ with }\sigma_2(r_1)=0.
	\end{align*}
	Combining $[OB,\sigma_2(r_1)=0]$ with $[FD,r_1,\omega_2]$ provides bounds on the degree in which assignments can be adjusted to $r_2$.
	
	\begin{lemma}
		\label{lem:FDOBt2}	
		In any IC mechanism, $\sigma_2(r_1)=0$ for some $r_1$ implies \[\gamma\geq \E_{\mu_2}[\hat{x}(r_1,\tilde{r}_2)]-\hat{x}(r_1,\O)\geq0.\]
	\end{lemma}
	An alternative way to look at this condition is that \[\mu_2[\hat{x}(r_1,1)-\hat{x}(r_1,\O)]+(1-\mu_2)
	[\hat{x}(r_1,0)-\hat{x}(r_1,\O)]\in[0,\gamma].\]
	In other words, the (weighted) average difference between $\hat{x}(r_1,r_2)$ and $\hat{x}(r_1,\O)$ for $r_2\in\{0,1\}$ must be in the interval $[0,\gamma]$. Thus, $\gamma$ (equivalently $c$) determines the scope that $P$ has to request no test and differentiate the assignments according to $r_2$. As the following result shows, when $c=0$, this scope completely closes, and no test in $t=2$ is incentive compatible if and only if the assignments do not condition on $r_2$.
	
	\begin{lemma}
		\label{lem:FDOBt2c0}
		Assume $c=0$. In any IC mechanism if $r_1$ is such that $\sigma_2(r_1)=0$ we must have $\hat{x}(r_1,0)=\hat{x}(r_1,1)=\hat{x}(r_1,\O).$
	\end{lemma}
	
	Now we turn our focus to $A$'s  testing decision in the first period when $\sigma_1=0$. If $A$ obeys, he gets $\E_{\mu_0}[v(\tilde{r}_1,\tilde{\mu}_2)]$ with probability $\pi$ and $\E_{\mu_0}[v(\O,\tilde{\mu}_2)]$ otherwise. In contrast, by testing $A$ expects to get $-c+\E_{\mu_0}[v(\tilde{r}_1,\tilde{\mu}_2)]$. Therefore, $A$ obeys $\sigma_1=0$ if and only if
	\begin{align*}
		&[OB,\sigma_1=0]&\quad\gamma\geq \E_{\mu_0}[v(\tilde{r}_1,\tilde{\mu}_2)]-\E_{\mu_0}[v(\O,\tilde{\mu}_2)].
	\end{align*} 
	Similar to the second period, $[OB,\sigma_1=0]$ and the full disclosure constraints put bounds on the scope that the principal has to vary $v(r_1)$ with $r_1$.
	\begin{lemma}
		\label{lem:FBOBt1}
		In any IC mechanism with $\sigma_1=0$ we must have
		\[\gamma\geq \E_{\mu_0}[v(\tilde{r}_1,\tilde{\mu}_2)]-\E_{\mu_0}[v(\O,\tilde{\mu}_2)]\geq 0.\]
	\end{lemma}
	
	\subsection{The optimal mechanism}
	
	The problem is formulated under the restriction $\sigma_1=0$. The proof of \autoref{prop:HDopt} shows that this restriction is indeed satisfied.
	
	Given a mechanism $(x,\sigma)$ we can express $P$'s expected payoff as
	\begin{align*}
W(x,\sigma):=\E_{x,\sigma}[\mathbb{1}[\hat{x}(\tilde{r}_1,\tilde{r}_2)=\tilde{\omega}_2]]-k\E_{\sigma_1=0}[\sigma_2(\tilde{r}_1)]
	\end{align*}

	where the expectation is taken over the distribution over public histories jointly induced by the distribution of states of the world, tests (free and mandatory), and allocations generated by $IC$ mechanisms.
		
	Hence, $P$'s program is given by
	\begin{align*}
		&\max_{(x,\sigma)\in\mathcal{M}}& &W(x,\sigma),\\
		& s.t.&\\
		&[OB,\sigma_1=0]&  &\gamma\geq \E_{\mu_0}[v(\tilde{r}_1,\tilde{\mu}_2)]-\E_{\mu_0}[v(\O,\tilde{\mu}_2)],  \\
		&[OB,\sigma_2(r_1)=1]& &\gamma\leq \E_{\mu_2}[\hat{x}(r_1,\tilde{r}_2)]\quad\forall r_1\text{ with }\sigma_2(r_1)=1,\\
		&[OB,\sigma_2(r_1)=0]& &\gamma\geq \E_{\mu_2}[\hat{x}(r_1,\tilde{r}_2)]-\hat{x}(r_1,\O)\quad\forall r_1\text{ with }\sigma_2(r_1)=0,\\
		&[FD,\omega_1=1]& &v(1,\rho)\geq v(\O,\rho), \\
		&[FD,\omega_1=0]& &v(0,1-\rho)\geq v(\O,1-\rho),\\
		&[FD,r_1,\omega_2=1]& &\hat{x}(r_1,1)\geq \hat{x}(r_1,\O)\quad\forall r_1\text{ with }\sigma_2(r_1)=0, \\
		&[FD,r_1,\omega_2=0]& &\hat{x}(r_1,0)\geq \hat{x}(r_1,\O)\quad\forall r_1\text{ with }\sigma_2(r_1)=0. 
	\end{align*}
	The following result characterizes the solution to $P$’s program when testing costs are intermediate for both parties. 
	
	\begin{proposition}
		\label{prop:HDopt}
		Consider the case in which \autoref{assump:interKappa} and \autoref{assump:interGamma}  hold. Then, the optimal mechanism is as follows: $\sigma_1=0$ and
		\begin{itemize}
			\item If $\kappa>\bar{\kappa}:=(1-\rho)/(1-\pi)$ and $\gamma\geq\bar{\gamma}$,  then the optimal testing policy is given by $\sigma_2(r_1)=\mathbb{1}[r_1=\O]$, while the optimal assignment policy is given by  
			\[\hat{x}(r_1,r_2)=\mathbb{1}[r_1=1]+\mathbb{1}[r_1\neq1]\mathbb{1}[r_2=1].
			\]
			\item If $\kappa\in(1-\rho,\bar{\kappa})$ or $\gamma<\bar{\gamma}$, then $\sigma^{RD1}$ is the optimal testing policy, while the optimal assignment policy is given by  
			\[\hat{x}(r_1,r_2)=\mathbb{1}[r_2=1].
			\]
			\item If $\kappa=\bar{\kappa}$, any of the two previous mechanisms is optimal.
		\end{itemize}
	\end{proposition}

	In other words, the previous result establishes that, for the region of interest (when \autoref{assump:interKappa}  holds  and the deviation in \autoref{def:deviation} is strictly profitable for $A$), $\sigma^{RD1}$ is indeed the \textit{optimal testing policy}. 
		
	The main intuition behind \autoref{prop:HDopt} is as follows. On one hand, when the deviation in \autoref{def:deviation} is profitable for $A$, preventing him from testing in the first period requires the continuation value $v(r_1,\mu_2)$  to be invariant to $r_1$. On the other hand, even when such deviation is not profitable, other deviations still constrain the design of an optimal mechanism. Specifically, $A$ may want to test after $r_1=1$ when $P$ finds that decision too costly.

	In particular, as in the baseline mechanism, the optimal mechanism must satisfy $\sigma_2(\O)=1$. Given the intermediate value of $\gamma$, $A$ is unwilling to test at $t=2$ after observing $\omega_1=0$. Thus, $P$ is left with two options following the report $r_1=1$: $(i)$ distort the efficient assignments by setting  $\hat{x}(1,r_2)=1$ for all $r_2$, thereby increasing $\hat{x}^*(1,0)=0$ to $\hat{x}(1,0)=1$; or $(ii)$ distort the baseline testing policy by setting $\sigma_2(1)=1>\sigma^*(1)=0$. Option $(i)$ requires $\gamma\geq\bar{\gamma}$ (i.e., the deviation in \autoref{def:deviation} is not profitable for $A$) to be incentive compatible,  and it is optimal when $\kappa$ is sufficiently high. Otherwise, option $(ii)$ is optimal. 

	\section{Discussion}	
	
Our model captures an informational friction in the problem of optimally assigning livers from deceased donors to patients in the waiting list. The state $\omega_t$ can be interpreted as the social cost of not allocating the liver to the patient, i.e., whether a patient must have priority for a transplant. The test is any procedure that reveals such social cost, which in practice is a function of the MELD score of all patients waiting. 

Patient testing costs $c$ represent, in a stylized manner, expenses such as travel to testing facilities and other private burdens not internalized by the principal, including physical or psychological discomfort and various logistical burdens. Principal costs $k$, in turn, comprise reimbursement expenditures 
as well as any portion of patient costs that the principal effectively internalizes.
Conceptually, the principal serves as a reduced-form representation of an allocation rule designer who also bears the social costs induced by testing, even when these costs are not directly borne by the agent.

Within this framework, the model provides a plausible rationale for the strategic timing of reports in deceased-donor liver allocation discussed in \autoref{sec:facts}. In fact, when the principal testing costs are intermediate, the baseline mechanism requires testing only at the end of the recertification window (\autoref{prop:baseline}). When this is the case and agent costs are low enough, \autoref{prop:gaming1} states that agents might find it advantageous to deviate from the baseline, test more often than advised, but report only the test results that increase their priority. This mirrors the empirical evidence in \cite{gordon2021chapter3}: results that raise a patient’s priority are typically reported before the deadline, while results that lower priority disproportionately accumulate near the recertification deadline.

Our model parsimoniously captures a plausible force behind the observed behavior and highlights a direction for mitigating the strategic timing of testing and disclosure within recertification groups (\autoref{prop:HDopt}) by allowing testing recommendations to depend on previously reported information. One interpretation of this result is that, while the current update schedule moves in the right direction, finer adjustments of recertification intervals in response to prior reports, such as shortening them after priority-increasing results, may help reduce strategic timing in testing. 

\subsection{Many agents}
The model analyzed above does not capture an important feature of the application under consideration: organs for transplantation are allocated among many patients. This omission is intentional. In \autoref{appen:manyagents}, we formally present an analogous model with a continuum of agents, and an aggregate assignment feasibility constraint. We assume that all agents are identical and follow the same strategies whenever the mechanism treats them symmetrically. We show that, when attention is restricted to a particular class of \textit{admissible} mechanisms, the problem of designing a mechanism in the many-agent setting reduces to the single-agent case. In other words, the problem admits a representative-agent formulation.

The admissibility requirements imposed on mechanisms are: anonymity (agents submitting identical reports receive the same testing recommendations and assignments), label-independence (policies do not depend on agents' labels, indices, or names), and null-insensitivity (no zero-measure subset of agents can be pivotal). 

\section{Conclusions}
	
	We studied the problem faced by a decision-maker who needs to incentivize (without the use of contingent payments) evidence acquisition and disclosure from an agent who is subject to moral hazard (i.e., the decision-maker does not observe his choice beyond what he reports). We do it in a framework with deterministic exogenous timing of choice, binary action and state spaces, and a fully revealing information acquisition technology.
	
	In this setting, we show that ignoring the agent's incentive concerns leads to  mechanisms that do not condition testing recommendations on specific test results. Moreover, decisions are solely based on payoff-relevant information. Once we take into account the agent's incentives, optimal testing policies tend to display a particular type of history dependence: they recommend further evidence acquisition when the current information supports the agent's favorite choice. Moreover, the optimal assignment may be contingent on payoff-irrelevant information.
	
	Interpreting the model in the context of deceased-donor liver transplantation, an informational friction in the MELD recertification process, arising from discretionary reporting between mandatory updates, can induce systematic  selection in the reporting of test results, just as the empirical evidence suggests. Our analysis indicates that result-dependent testing schedules may better align patient behavior with the designer’s objectives, thereby mitigating the non-disclosure of voluntarily acquired information. Exploring environments with longer time horizons, stochastic organ arrivals, and richer state and action spaces remains a promising avenue for future research.
		
	Finally, the substantive insights from our analysis highlight an important consideration for decision-makers who rely on delegated information acquisition. Ignoring agents' preferences can give rise to opportunistic behavior, but such conduct can be mitigated through carefully designed schedules. In particular, conditioning evidence requests on previously reported information may serve as a useful tool in other environments as well. 
	
	\pagebreak
	
	\section*{References}
	\bibliography{tests}

	\newpage
	
	\appendix 
	
	\section{Omitted Proofs}
	
	\begin{proof}[Proof of \autoref{prop:baseline}]
		We proceed by backward induction. At the end of period $2$ but before assignments are realized, efficient assignments deliver an expected payoff of 
		\[
		\max\{\mu_2(h_2),1-\mu_2(h_2)\}
		\]
		to $P$. 
		 	
		Regarding testing choices, given a fixed $r_1$, $P$ can test in $t=2$ and get to fully adjust the action to the state at a cost of $k$. Alternatively, she can hope to obtain a result for free (which occurs with probability $\pi$). Otherwise, she must base the decision solely on the information held at the beginning of the period, i.e., $\mu_2(r_1)$. Therefore, testing in $t=2$ is optimal if and only if
		\begin{align*}
		\mu_2(h_1)(1)+(1-\mu_2(h_1))(1)-k\geq&\\ \pi[\mu_2(h_1)(1)+(1-\mu_2(h_1))(1)]&+(1-\pi)\max\{\mu_2(h_1),(1-\mu_2(h_1))\}			
		\end{align*}
		
		or $\kappa\leq 1-\max\{\mu_2(r_1),1-\mu_2(r_1)\}$.\footnote{Recall that $\kappa=k/(1-\pi)$ is the \textit{effective} cost of testing for the principal.} If $\kappa>1/2$, this condition cannot be satisfied. If $\kappa=1/2$, testing is optimal only when $\mu_2(r_1)=1/2$. Finally, for $\kappa<1/2$, testing is optimal if and only if the belief is sufficiently close to $1/2$, i.e., $\mu_2(r_1)\in[\kappa,1-\kappa]$.
				
		Observe that $r_1 \in \{0,1,\O\}$. On one hand, when $r_1\in\{0,1\}$ we have that $\mu_2(r_1)\in\{\rho,1-\rho\}$. Thus, \[\max\{\mu_2(r_1),1-\mu_2(r_1)\}=\rho\] implies that testing in $t=2$ is optimal if and only if $\kappa\leq 1-\rho$. 
		
		On the other hand, if $r_1=\O$, the belief is equal to $\mu_2(\O)=\rho\mu_0+(1-\rho)(1-\mu_0)$. Note that $\mu_2(\O)\geq 1/2$ is equivalent to $\mu_0\geq 1/2$. Thus, if $\mu_0\geq 1/2$ then $\max\{\mu_2(\O),1-\mu_2(\O)\}=\mu_2(\O)$, implying that testing is optimal if and only if $\kappa\leq 1-\mu_2(\O)$. Similarly, if $\mu_0< 1/2$ then $\max\{\mu_2(\O),1-\mu_2(\O)\}=1-\mu_2(\O)$, thus testing is optimal if and only if $\kappa\leq \mu_2(\O)$. Hence, testing is optimal after $r_1=\O$ if and only if $\kappa\leq\min\{\mu_2(\O),1-\mu_2(\O)\}$.
		
		Note that the cut-off for testing is higher after $r_1=\O$ than after $r_1\in\{0,1\}$, since for all $\mu_0\in (0,1)$ we have $\min\{\mu_2(\O),1-\mu_2(\O)\}>1-\rho$.
		
		In conclusion, if $\kappa\in(1-\rho,\min\{\mu_2(\O),1-\mu_2(\O)\}]$, it is optimal to test in $t=2$ if and only if no test result was obtained in $t=1$. Moreover, if 
		$\kappa\leq 1-\rho$ it is always optimal to test in $t=2$, while if $\kappa>\min\{\mu_2(\O),1-\mu_2(\O)\}$ it is never optimal to test in $t=2$ (regardless of the history).
		
		In all cases, it is never optimal for $P$ to test in $t=1$.		
	\end{proof}
	
		\begin{proof}[Proof of \autoref{lem:testObedt2}]
			Let $m^*:=(\sigma^*,\hat{x}^*)$ be the baseline mechanism.  Firstly, in the disclosure stage, in any $t$ and $h_t$ with $\sigma_t^*(h_t)=0$ it is optimal for $A$ to set $r^{m^*}_t(h_t,e_t,\omega_t)=\mathbb{1}[\omega_t=1]$. 
		
		Secondly, in the testing stage, $r_1=1$ implies: $(i)$ $\sigma^*_1(1)=0$, and $(ii)$ $\hat{x}^*(1,\O)=1$. Hence, $e^{m^*}_2(1)=0$ whenever $c>0$. Because under the corresponding interim belief, $\mu_2=\rho>1/2$, $A$ gets his preferred allocation and any unfavorable free evidence can be concealed  (because a test result is not requested in this situation).

		Moreover, $r_1=0$ implies: $(i)$ $\sigma_2^*(0)=0$, and $(ii)$ $\hat{x}^*(0,\O)=0$. Therefore, $e^{m^*}_2(0) =1$ if and only if	
	\[
	-c+[(1-\rho)(1)+\rho (0)]\geq \pi [(1-\rho)(1)+\rho (0)]+(1-\pi)(0),
	\]
	equivalently $\gamma\leq 1-\rho$. Thus, $A$ follows $\sigma_2^*(1)=\sigma_2^*(0)=0$ if and only if  $\gamma>1-\rho$.

	Furthermore, $r_1=\O$ implies: $(i)$ $\sigma_2^*(\O)=1$, and $(ii)$ $\hat{x}^*(\O,\O)=\hat{x}^*(\O,0)=0$. Hence, $A$ is indifferent between $r_2=0$ and $r_2=\O$. Thus, $e^{m^*}_2(\O)=1$ if and only if $\gamma\leq \mu_2$, where $\mu_2$ denotes $A$'s private belief at the beginning of period 2. Note that $r_1=\O$ can be the outcome of either absence of evidence or concealed evidence. In the first case, $\mu_2=\mu_2(\O)$; in the second, $\mu_2=1-\rho$. Because $\mu_2(\O)>1-\rho$,  we have three cases: First, if $\gamma\leq 1-\rho$,  $e^{m^*}_2(\O)=1$. Second, if $\gamma>\mu_2(\O)$, $e^{m^*}_2(\O)=0$. Finally, for intermediate testing costs $\gamma\in(1-\rho,\mu_2(\O)]$, $e^{m^*}_2(\O)=1$ if and only if no result went unreported in the previous period.
	
	\end{proof}
	
	\begin{proof}[Proof of \autoref{prop:gaming1}]

Recall that $m^*=(\sigma^*,\hat{x}^*)$ denotes the baseline mechanism. If $e^{m^*}_1=\sigma^*_1=0$, $A$ can expect a free test is $t=1$ with probability $\pi$. The disclosure of the free test result is given by $r_1^{m^*}(1,\omega_1)=\mathbb{1}[\omega_1=1]$. When $r_1=1$ we have that $\hat{x}^*(1,\O)=1$, thefore $A$ receives his preferred allocation with probability $1$ after $r_1=1$. Hence, $e^{m^*}_2(1)=0$. 

In contrast, by \autoref{lem:testObedt2},when $\gamma\in(1-\rho,\mu_2(\O)]$, $A$ is not willing to follow $\sigma_2(\O)=1$ after observing $\omega_1=0$. Thus, after observing $\omega_1=0$ and reporting $r_1=\O$, $A$ expects an allocation of 1 if and only if a free test reveals $\omega_2=1$, which occurs, from $A$'s point of view, with probability $\pi(1-\rho)$. If $r_1=\O$ is a result of no free test, which occurs with probability $1-\pi$, $A$ will be willing to follow $\sigma_2(\O)=1$ and can expected a payoff of $-c+\mu_2(\O)$, bacause $\hat{x}^*(\O,r_2)=\mathbb{1}[r_2=1]$. Thus if $e^{m^*}_1=\sigma^*_1=0$, $A$ can expect a payoff of 
\[-c(1-\pi)+\pi[\mu_0+(1-\mu_0)\pi(1-\rho)]+(1-\pi)\mu_2(\O).\]

As a result, the $A$ (strictly) prefers the deviation in \autoref{def:deviation} if and only if \[-c+\mu_0+(1-\mu_0)\pi(1-\rho)>-c(1-\pi)+\pi[\mu_0+(1-\mu_0)\pi(1-\rho)]+(1-\pi)\mu_2(\O).\] Re-arranging this expression and using $\mu_2(\O)=\mu_0\rho+(1-\mu_0)(1-\rho)$ we get the desired result.

	\end{proof}

	\begin{proof}[Proof of \autoref{lem:implemEffAssign}]
	Let $m=(\sigma^{RD1},\hat{x}^*)$ be the mechanism that pairs the testing policy $\sigma^{RD1}$ with the efficient assignments $\hat{x}^*$. 	We know that $e_2^{m}(0)=0$ and $e_2^{m}(1)=1$ if and only if  $\gamma\in(1-\rho,\rho]$.

	Note that $\sigma^{RD1}(\O)=1$. Thus, $e^m_2(\O)$ if $-c+\mu_2\geq \pi\mu_2$ equivalently $\gamma\leq\mu_2$ (where $\mu_2$ is $A$'s private belief at the beginning of period 2). Therefore, if $r_1=\O$ is the result of withheld information then $e_2^m(\O)=0$ when $\gamma\in(1-\rho,\rho]$. But if $r_1=\O$ is the result of no available information, then $e_2^m(\O)=1$ if $\gamma\leq \mu_2(\O)$. A a result, provided no information was withheld in $t=1$, $A$ follows all recommendations in the second period if and only if $\gamma\in(1-\rho,\mu_2(\O)]$.
		
	Next, we check whether $r_1=\omega_1$ is on $A$'s interest. If $A$ observes $\omega_1=0$, reporting $r_1=0$ leads an expected payoff of $\pi(1-\rho)$ (since no test will be requested, which he follows as long as $\gamma>1-\rho$, and unless $\omega_2=1$ is observed for free the allocation is $0$). When  $\gamma\in(1-\rho,\rho]$ reporting $r_1=\O$ yields an expected payoff of $\pi(1-\rho)$, since he will not follow the recommendation to test and the only scenario in which the allocation is not zero is after observing $\omega_2=1$ for free. Thus, $A$ is willing to truthfully report $r_1=0$.
		
	If $A$ observes $\omega_1=1$ and reports $r_1=1$ he can expect a payoff of $-c+\rho$. If he instead reports $r_1=\O$, he gets $-c+\rho$ (since a test result is requested and the probability that the test result is high equals $\rho$ given the observation $\omega_1=1$). Thus, $A$ is also willing to truthfully report $r_1=1$.

	Thus, when $\gamma\in(1-\rho,\mu_2(\O)]$ we have that $A$ fully reports any information in either period and follows all recommendations in $t=2$. 
		
	We are only left with the testing decision in the first period when $\gamma\in(1-\rho,\mu_2(\O)]$. 
	
	In $t=1$ $A$ tests if
	\[-c+\mu_0 (-c+\rho)+(1-\mu_0)\pi(1-\rho)\geq \pi [\mu_0 (-c+\rho)+(1-\mu_0)\pi(1-\rho)]+(1-\pi)[-c+\mu_2(\O)],\]
	equivalently,
	\[\gamma\leq\frac{-(1-\pi)(1-\rho)(1-\mu_0)}{\pi+\mu_0(1-\pi)},\]
		
	which can never hold since the left-hand side is strictly positive and the right is strictly negative. Therefore, obedience to $\sigma_1=0$ is also guaranteed.
	\end{proof}

\begin{proof}[Proof of \autoref{lem:lowerbound}]
	
	If $\sigma_2(1)=0$ then disclosure of $\omega_2=0$ after $r_1=1$ requires $\hat{x}(1,0)\geq \hat{x}(1,\O)$ which forces either $\hat{x}(1,0)$  or $\hat{x}(1,\O)$ away from their efficient levels (0 and 1, respectively). 
	
	Moreover, if $\sigma_2(\O)=0$, disclosure of $\omega_2=0$ after $r_1=\O$ requires $\hat{x}(\O,0)\geq \hat{x}(\O,\O)$. Since $\mu_0>1/2$ is equivalent to $\mu_2(\O)>1/2$, the assignments need to be distorted from their efficient levels (0 and 1, respectively). If $\mu_2(\O)<1/2$, we can still have the efficient assignments  $\hat{x}(\O,0)=\hat{x}(\O,\O)=0$. However, $A$ would not follow $\sigma_2(\O)=0$. By testing $A$ gets an expected payoff of $-c+\mu_2(\O)$, while by not testing $A$ gets $\pi \mu_2(\O)$. Hence, $A$ follows $\sigma_2(\O)=0$ if and only if $\gamma\geq \mu_2(\O)$, which cannot hold when  $\gamma\in(1-\rho,\mu_2(\O))$. 	
\end{proof}

	\begin{proof}[Proof of \autoref{prop:revelation}]
	First, note that for a given mechanism, $A$ faces a decision problem. Whenever $A$ is indifferent between two actions, we adopt the tie-breaking rule that he follows the test recommendation and discloses the test result. Thus, The actions compatible with $A$'s rationality are uniquely pinned down by the mechanism, under this tie-breaking.
	
	Consider an alternative mechanism $M''=(\hat{x}'',\sigma'')$ such that
	\begin{align*}
		\hat{x}''(h_t,r_t):=
		\left\{ 
		\begin{tabular}{c c c}
			$\hat{x}(h_t,\omega_t)$ & if & $r^{M}_t(h_t,e_t,\omega_t)=\omega_t$,\\
			$\hat{x}(h_t,\O)$ & if & $r^{M}_t(h_t,e_t,\omega_t)=\O.$
		\end{tabular}\right.
	\end{align*}
	and 
	\begin{align*}
		\sigma''_t(h_t):=
		\left\{ 
		\begin{tabular}{c c c}
			$\sigma_t(h_t)$ & if & $e^M_t(h_t)=\sigma_t(h_t),$ \\
			$1-\sigma_t(h_t)$ & otherwise. &
		\end{tabular}\right.
	\end{align*}

	On one hand, the alternative assignment policy gives an agent with test result $\omega_t$ the same allocation than the original policy, regardless of the report that $A$ was optimally choosing under the original policy. On the other hand, the alternative testing policy recommends to test at given history if and only if $A$ was testing at that same history under the original policy (regardless of what the original recommendation was). Hence, the mechanism $M''$ keeps the same on-path actions and outcomes as $M$. 
		
	Now, we adjust $M''$ such that any deviation from the agent is punished with his least preferred action without requiring any additional tests, effectively ending the relationship between the parties. Consider the mechanism $M'=(x',\sigma')$ such that
	\begin{align*}
		\hat{x}'(h_t,r_t):=
		\left\{ 
		\begin{tabular}{c c c}
			$\hat{x}''_t(h_t,\omega_t)$ & if & $r^{M''}_{t'}(h_{t'},e_{t'},\omega_{t'})=\omega_{t'} \ \forall t'\leq t \text{ with } \sigma''_{t'}(h_{t'})=1$\\
			$0$ & otherwise
		\end{tabular}\right.
	\end{align*}
	and 
	\begin{align*}
		\sigma'_t(h_t):=
		\left\{ 
		\begin{tabular}{c c c}
			$\sigma''_t(h_t)$ & if & $r^{M''}_{t'}(h_{t'},e_{t'},\omega_{t'})=\omega_{t'} \ \forall t'\leq t \text{ with } \sigma''_{t'}(h_{t'})=1$\\
			$0$ & otherwise. &
		\end{tabular}\right.
	\end{align*}
	On-path $M'$ coincides with $M''$. Moreover, they only make $A$ more willing to follow the recommendations of the mechanism. Thus, if there were no profitable deviations for $A$ under $M''$, there are none under $M'$.
	\end{proof}
	
	\begin{proof}[Proof of \autoref{lem:FDOBt2}]
		Since any $\mu_2\in(0,1)$ we have that $[FD,r_1,\omega_2=1]$ is equivalent to $\mu_2 x(r_1,1)\geq \mu_2 \hat{x}(r_1,\O)$. Likewise,  $[FD,r_1,\omega_2=0]$ is equivalent to $(1-\mu_2) \hat{x}(r_1,0)\geq (1-\mu_2) \hat{x}(r_1,\O)$. Adding the two we get $\E_{\mu_2}[\hat{x}(r_1,\tilde{r}_2)]\geq \hat{x}(r_1,\O)$. $[OB,\sigma_2(r_1)=0]$ provides an upper bound on $\E_{\mu_2}[\hat{x}(r_1,\tilde{r}_2)]-\hat{x}(r_1,\O)$.
	\end{proof}
	
		\begin{proof}[Proof of \autoref{lem:FDOBt2c0}]
		$c=0$ is equivalent to $\gamma=0$. By lemma \ref{lem:FDOBt2} we have $\E_{\mu_2}[\hat{x}(r_1,\tilde{r}_2)]=\hat{x}(r_1,\O)$. Suppose  $\hat{x}(r_1,1)>\hat{x}(r_1,0)$. Then, $\hat{x}(r_1,\O)>\hat{x}(r_1,0)$ which violates  $[FD,r_1,\omega_2=0]$.
		Similarly, if  $\hat{x}(r_1,1)<\hat{x}(r_1,0)$ then $\hat{x}(r_1,\O)>\hat{x}(r_1,1)$ which violates  $[FD,r_1,\omega_2=1]$. Hence, the only possibility left is $\hat{x}(r_1,0)=\hat{x}(r_1,1)=\hat{x}(r_1,\O)$.	
	\end{proof}	
	
		\begin{proof}[Proof of \autoref{lem:FBOBt1}]
		Given that $\mu_0\in(0,1)$ we have that $[FD,\omega_1=1]$ is equivalent to $\mu_0 v(1,\rho)\geq \mu_0 v(\O,\rho)$ and $[FD,\omega_1=0]$ is equivalent to $(1-\mu_0) v(0,1-\rho)\geq (1-\mu_0) v(\O,1-\rho)$. Adding both expressions we get $\E_{\mu_0}[v(\tilde{r}_1,\tilde{\mu}_2)]-v(\O,\tilde{\mu}_2))]\geq 0$. $[OB,\sigma_1=0]$ provides the upper bound.
	\end{proof}

		\begin{lemma}
			\label{lem:optempt}
			Assume $\kappa\leq \min\{\mu_2(\O),1-\mu_2(\O)\}$ and $\gamma\leq \mu_2(\O)$. Then, the optimal mechanism must have $\sigma_2(\O)=1$ and $\hat{x}(\O,r_2)=\mathbb{1}[r_2=1]$.
		\end{lemma}
	
	\begin{proof}[Proof of \autoref{lem:optempt}]
	 First, note that it is IC. $[OB,\sigma_2(\O)=1]$ requires $\gamma\leq \mu_2(\O)$, which is assumed to hold. A change of the policy to $\sigma_2(\O)=0$ violates $[OB,\sigma_2(\O)=0]$ when paired with the same assignment policy. In order to make $[OB,\sigma_2(\O)=0]$ hold, there are two options. The first one is to set $\hat{x}(\O,\O)=1$, while keeping the other assignments fixed. The second one is to set  $\hat{x}(\O,1)=0$, while keeping the other assignments fixed. Neither case represents an increase in allocational efficiency (since when $\sigma_2(\O)=1$ the report $r_2=\O$ is not sent with positive probability). Finally, $\kappa\leq \min\{\mu_2(\O),1-\mu_2(\O)\}$  implies that testing after $r_1=\O$ is optimal.
	\end{proof}
	
	\begin{proof}[Proof of \autoref{prop:HDopt}]
		By \autoref{lem:optempt}, the optimal mechanism has $\sigma_2(\O)=1$ and $\hat{x}(\O,r_2)=\mathbb{1}[r_2=1]$. 
		
		Full disclosure in $t=2$ required $\hat{x}(r_1,0)\geq \hat{x}(r_1,\O)$ for all $r_1$ with $\sigma_2(r_1)=0$. Suppose these constraints bind for $r_1\in\{0,1\}$.\footnote{Setting $\hat{x}(r_1,0)>\hat{x}(r_1,\O)$ can only decrease $W$. Moreover, they only make $[OB,\sigma_2(r_1)=0]$ harder to be satisfied.} $[OB,\sigma_2(r_1)=0]$ becomes $\gamma\geq \mu_2[\hat{x}(r_1,1)-\hat{x}(r_1,0)]$. Thus, if  $\hat{x}(r_1,1)=1$ and $\hat{x}(r_1,0)=0$ (as in the efficient assignments), $[OB,\sigma_2(r_1)=0]$ becomes $\gamma\geq \mu_2$. Since we assumed $\gamma>1-\rho$ it follows that $[OB,\sigma_2(0)=0]$ is satisfied. Moreover, $v(0)=v(0,1-\rho)=\pi(1-\rho)>1-\rho-c=v(\O,1-\rho)$, thus $[FD,\omega_1=0]$ is satisfied.  However, since $\gamma\leq \mu_2(\O)< \rho$, this recommendation is not incentive compatible with those assignments. $P$ could decrease $\hat{x}(1,1)$ from 1 to 0. Then $[OB,\sigma_2(1)=0]$ becomes $\gamma\geq0$, which is satisfied. However, this implies that $v(1)=0<-c+\rho=v(\O,\rho)$, which violates $[FD,\omega_1=1]$. Therefore, this adjustment does not make the mechanism incentive compatible. There are two ways to recover incentive compatibility:
		
		$(i)$  Set $\hat{x}(1,r_2)=1$ for all $r_2$. This clearly satisfies obedience and full disclosure in $t=2$. This also satisfies $[FD,\omega_1=1]$ since $v(1)=1>-c+\rho=v(\O,\rho)$. Given full disclosure $[OB,\sigma_1=0]$ becomes $\mu_0 (1)+(1-\mu_0)\pi(1-\rho)-(\mu_2(\O)-c)\leq \gamma$, equivalently, $\gamma\geq \bar{\gamma}=(1-\rho)[\mu_0-(1-\mu_0)(1-\pi)]/\pi$.\footnote{Note that $\bar{\gamma}$ appeared in \autoref{prop:gaming1} as the threhold for which if $\gamma<\bar{\gamma}$ then $A$ strictly prefers to deviate from the baseline testing recommendations.} 
		
		The expected payoff under this mechanism equals 
		\begin{align*}
		%&\pi^2\mu_0\rho+\pi(1-\pi)\mu_0\rho\\
		%&+\pi^2(1-\mu_0)(1-\rho)+\pi^2(1-\mu_0)\rho+\pi(1-\pi)(1-\mu_0)\rho\\
		%&+(1-\pi)\mu_2(\O)(1-k)+(1-\pi)(1-\mu_2(\O))(1-k)\\
		%&=\pi\mu_0\rho+\pi^2(1-\mu_0)+\pi(1-\pi)(1-\mu_0)\rho+(1-\pi)(1-k)=\\
		&\pi\mu_0\rho+\pi(1-\mu_0)[\pi+(1-\pi)\rho]+(1-\pi)(1-k).
		\end{align*}

		$(ii)$ Change the testing recommendation to $\sigma_2(1)=1$ while maintaining the efficient assignments. This change does not affect any of the other constraints. 
		
		The expected payoff under this mechanism equals 
		\begin{align*}
	%&\pi\mu_0\rho(1-k)+\pi\mu_0(1-\rho)(1-k)\\
	%&+\pi^2(1-\mu_0)(1-\rho)+\pi^2(1-\mu_0)\rho+\pi(1-\pi)(1-\mu_0)\rho\\
	%&+(1-\pi)\mu_2(\O)(1-k)+(1-\pi)(1-\mu_2(\O))(1-k)\\
	%&=\pi\mu_0(1-k)+\pi(1-\mu_0)[\pi+(1-\pi)\rho]+(1-\pi)(1-k)=\\
	&(1-k)[\pi\mu_0+1-\pi]+\pi(1-\mu_0)[\pi+(1-\pi)\rho].
		\end{align*}

		Therefore, the optimal mechanism option $(i)$ requires $\gamma\geq \bar{\gamma}$ and it is strictly preferred to option $(ii)$ if and only if  $\kappa> \bar{\kappa}:=(1-\rho)/(1-\pi)$. Otherwise option $(ii)$ is optimal.

		Observe that $\hat{x}(r_1,r_2)=\mathbb{1}[r_2=1]$ coincides with the efficient assignment for all reports that are sent with positive probability under the optimal mechanism. This follows because the optimal assignment policy only distorts $\hat{x}^*(1,\O)$ and $\hat{x}^*(\O,\O)$. Additionally,  $\sigma_2(\O)=\sigma_2(1)=1$ satisfies obedience and full disclosure when assumptions \ref{assump:interKappa} and \ref{assump:interGamma} hold. Therefore, no mechanism with $\sigma_1=0$ can deliver strictly higher payoffs to $P$.

		Finally, to study the optimality of $\sigma_1=0$, note that only $r_1=0$ leads to no testing. Therefore, the only reason to test in $t=1$ would be to reduce potential distortions following that report. However, there are no distortions after $r_1=0$, since \[\mu_2(r_1=0,r_2=\O)=1-\rho<1/2.\]
		Thus, testing at $t=1$ can only increase costs and produce no gains in terms of efficiency, which implies that $\sigma_1=0$ is optimal. 				
	\end{proof}
	
	\pagebreak
	
\section{Many agents}
\label{appen:manyagents}

\renewcommand{\thedefinition}{\Alph{section}.\arabic{definition}}
\renewcommand{\theassumption}{\Alph{section}.\arabic{assumption}}
\renewcommand{\theproposition}{\Alph{section}.\arabic{proposition}}
\renewcommand{\thelemma}{\Alph{section}.\arabic{lemma}}
\renewcommand{\theequation}{\Alph{section}.\arabic{equation}}

\setcounter{equation}{0}
\setcounter{definition}{0}
\setcounter{assumption}{0}
\setcounter{proposition}{0}
\setcounter{lemma}{0}

\subsection{Model}
Let $N:=[0,1]$ be the set of agents, endowed with the Borel sigma-algebra $\mathcal{B}(N)$ and the Lebesgue measure $\lambda$. Let  $i\in N$ denote a typical agent. 

There is a measurable set $S$ of identical items to be allocated among the agents. We assume that the supply of items is not enough to serve all agents, i.e., $0<\lambda(S)<\lambda(N)=1$. At the end of the second period, the principal distributes the mass $\lambda(S)$ of items among the agents.  Formally, an \textit{assignment} $x:N\to\{0,1\}$ is a measurable function for which the condition
\begin{align}
	\label{condi:aggregate}
	\int_Nx(i)d\lambda(i)\leq\lambda(S)
\end{align}

is satisfied.\footnote{Equivalently, an \textit{assignment} is a set $N^{priority}\in\mathcal{B}(N)$ such that $\lambda(N^{priority})\leq\lambda(S)$. An agent $i$ gets a fraction of the item if and only if $i\in N^{priority}$.}

In period $t$, the state is given by $\omega_t:N\to \{0,1\}$ where $\omega_t(i)\in\{0,1\}$ denotes the state for agent $i$ in that period. Each agent's states are independently and identically distributed. The typical state for an agent $i$ in period $t$ is denoted by $\omega$.  
The state for each agent evolves independently with transition probabilities given by $	\Pr[\tilde{\omega}_{t+1}(i) = \omega \mid \tilde{\omega}_{t}(i) = \omega]=\rho$. Moreover, the common prior belief about the initial state for each agent is given by  $\mu_0\in(0,1)$.

Regarding tests, we assume that all agents face the same testing cost $c$ and each test generates a cost $k$ to the principal. If $i$ tests in $t$ he privately observes $\omega_t(i)$. If agent $i$ decides not to test in $t$, he can still observe his own state at $t$ with probability $\pi\in(0,1)$. Let $e_t(i)\in\{0,1\}$ denote $i$'s testing decision in period $t$. If $i$ observes no test result in $t$ then $r_t(i)=\O$, if he has observed that $\omega_t(i)=\omega$ then $r_t(i)\in\{\O,\omega\}$.  Testing and reporting decisions are made simultaneously and privately by each agent. 

The \textit{ex-post} payoffs are as follows. Given an assignment $x$ and testing decisions $(e_1(i),e_2(i))$ for each agent $i$, $P$ receives
\[\int_N\mathbb{1} [\omega_2(i)=x(i)]d\lambda(i)-k\int_N [e_1(i)+e_2(i)]d\lambda (i),\]
while each agent $i$ receives
\[x(i)-c[e_1(i)+e_2(i)].\]

We assume that the report sent by agent $i$ in period $t$, $r_t(i)\in\{\O,0,1\}$, is only observed by $i$ and $P$. Therefore, a \textit{private history for agent $i$} at the end of period $t$ corresponds to $h_t(i):=(r_1(i),\ldots,r_t(i))$. Let $\mathcal{H}_t:=\{\O,0,1\}^t$ denote the set of possible private histories. While a \textit{collective history} for $P$ at the end of $t$, denoted $h_t:=\{h_t(i)\}_{i\in N}$, collects the private histories for all agents.

The principal commits to a mechanism at the beginning of the first period. A mechanism $m:=(\sigma,\hat{x})$ is the pairing of a \textit{testing policy} and an \textit{assignment policy} defined as follows:

\begin{itemize}
	\item A \textit{testing policy} $\sigma:=\{ (\sigma_1(i),\sigma_2(i))   \}_{i\in N}$ gives a testing recommendation to each agent in each period as a function of the collective history. In particular, $\sigma_1(i)\in\{0,1\}$ and $\sigma_2(i)$ maps $h_1$ into $\{0,1\}$. Let $\sigma_2(i;h_1)$ denote the testing recommendation for agent $i$ when the public history at the end of  the first period is $h_1$.
	
	\item 	An \textit{assignment policy} $\hat{x}$ maps the collective history at the end of period $2$, $h_2$, to an assignment $x$.  Let $\hat{x}(i;h_2)$ denote the assignment received by $i$ when the collective history at the end of the second period is $h_2$. We require $\hat{x}(i;h_2)$  to be measurable in $i$.
	
\end{itemize}

\subsection{Additional assumptions}

Having described the setting, we introduce additional assumptions necessary to rule out pathological cases and to properly address the issues arising from having a continuum of agents.

First, we restrict attention to a particular class of the mechanisms. Motivated by the fact that all agents are \textit{ex-ante} identical, it is natural to focus on mechanisms that treats all agents symmetrically. 	In other words, we impose that mechanisms can only differentiate agents by their reports. 	

\begin{definition}[Anonymous mechanisms]
	A mechanism is said to be \textit{anonymous} if for all  $i,j\in N$, all $h_1$, and all $h_2$:
	\begin{itemize}
		\item $\sigma_1(i)=\sigma_1\in\{0,1\}$,
		\item  $r_1(i)=r_1(j)$ implies $\sigma_2(i;h_1)=\sigma_2(j;h_1)$, and
		\item $(r_1(i),r_2(i))=(r_1(j),r_2(j))$ implies $\hat{x}(i;h_2)=\hat{x}(j;h_2)$.
	\end{itemize}
\end{definition}

Second, we interpret the indexes on the set of agents as carrying no meaning. In other words, changing the indexes of the agents should be inconsequential for the mechanism. In order to formalize this  \textit{desideratum},  define a relabeling $\phi:N\to N$ as a measurable-bijection that preserves the measure of every  subset, i.e., $\lambda(\phi^{-1}(B))=\lambda(B)$ for all $B\in \mathcal{B}(N)$.\ Let $\Phi$ be the set of all such bijections. For a \textit{collective} history $h_t$ and relabeling $\phi$ define the \textit{private} history for the ``new" agent $i$ as $(\phi\cdot h_t)(i):=h_t(\phi^{-1}(i))$.

\begin{definition}[Label-invariant mechanisms]
	A mechanism is said to be label-invanriant if for every $\phi\in\Phi$ and every collective history $h$, the following equalities hold for $\lambda$-a.e. agent $i\in N$:
	\begin{itemize}
		\item $\sigma_2(\phi(i);\phi\cdot h_1)=\sigma_2(i;h_1)$, and
		\item $\hat{x}(\phi(i);\phi\cdot h_2)=\hat{x}(i;h_2)$.
	\end{itemize}
\end{definition}
Note that $\sigma_1$ is trivially label-invariant under anonymity.

Third, we focus on mechanisms fow which no $\lambda$-measure zero subset of agents is pivotal. 

\begin{definition}[Null-insensitive mechanisms]
	A mechanism is said to be null-insensitive if for every two \textit{collective} histories $h_t$ and $h'_t$ such that 
	$\lambda(\{i\in N:h_t(i)\neq h'_t(i)\})=0$  and for $\lambda$-.a.e. $i\in N$:
	\begin{itemize}
		\item $\sigma_2(i;h_1)=\sigma_2(i;h'_1)$, and 
		\item $\hat{x}(i;h_2)=\hat{x}(i;h'_2)$. 
	\end{itemize}
\end{definition}
The following assumption is maintained.
\begin{assumption}[Admissible mechanisms]
	\label{assump:admissibility}
	Mechanisms are required to be anonymous, label-invariant, and null-insensitive. 
\end{assumption}

Additionally, since we are restricting the principal to use symmetric mechanisms. it is natural to focus on the case in which almost every agent employ the same testing and reporting strategies. Formally, 
\begin{assumption}[Symmetric agents' strategies]
	\label{assump:symmetric}
	For $\lambda$-a.e. agent the strategy is the same.
\end{assumption}

Finally, since the law of large numbers can not be directly invoked with a continuum of random variables, we  make the following technical assumption.

\begin{assumption}[Exact aggregation]
	\label{assump:exactaggregation}
	For each $t$ and each possible private history $h\in \mathcal{H}_t$ 
	\[
	\lambda\big(\{i\in N:\ h_t(i)=h\}\big)
	\quad\text{is }\text{ almost surely constant}.
	\]
\end{assumption}
The almost sure statement in the previous assumption are with respect to the probability measure governing states and availability of free tests. This assumption can be microfounded by extending the description of the model and using the exact law of large numbers for a continuum of agents in \citep{sun2006exact}. Motivated by parsimony, we directly assume this result.

\subsection{Result: Reduction to the one-agent case}
Let 
\[N_t(h):=\{i\in N:h_t(i)=h\}\]
be the set of agents with a \textit{private} history at the end of period $t$ given by $h\in\mathcal{H}_t$. Thus, $\{N_t(h)\}_{h\in\mathcal{H}_t}$ is a partition of $N$ according to \textit{private} histories. Additionally, let
\[\nu_t:=(\lambda(N_t(h)))_{h\in\mathcal{H}_t}\in\Delta(\mathcal{H}_t),\]
be the empirical distribution of \textit{private} histories at $t$.\footnote{ $\Delta(\mathcal{H}_t)$ denotes the simplex over $\mathcal{H}_t$.} Let $\nu_t(h)$ for $h\in\mathcal{H}_t$ be the typical entry of the vector $\nu_t$ and $\nu:=(\nu_1,\nu_2)$.

First, we establish that admissible mechanisms can only condition recommendations and assignments to a given agent on his own \textit{private} history and on $\nu$.
\begin{lemma}
	\label{lemma:equiv}
	If two possible \textit{collective} history profiles $(h_1,h_2)$ and $(h'_1,h_2')$ induce the same $\nu$, then there exist functions $s:{\O,0,1}\to\{0,1\}$ and $\chi:\{\O,0,1\}\to\{0,1\}$ such that for an admissible mechanism $m=(\sigma,\hat{x})$ and $\lambda$-a.e. $i\in N$:
	\begin{itemize}
		\item $\sigma_2(i;h_1)=\sigma_2(i;h'_1)=s(r_1(i))$, and
		\item $\hat{x}(i;h_2)=\hat{x}(i;h'_2)=\chi(r_1(i),r_2(i))$.
	\end{itemize} 
\end{lemma}
Note that the histories do not need to be the result of the equilibria in the game induced by a given mechanism, they only need to be possible. Althought the sets $N_t(h)$ are endogenous, the restrictions derived in the previous lemma are \textit{ex-ante} constraints on the admissible mechanisms class. They apply to every conceivable \textit{collective} history profile (and therefore do not depend on equilibrium strategies). 
\begin{proof}
	First, for any anonymous mechanism:
	\begin{itemize}		
		\item 	All agents in 	$N_1(r_1)$ receive the same testing recommendation $s_2(r_1)\in\{0,1\}$  in the second period. In other words, $\sigma_2(i;h_1)=s_2(h(i))$ for all $i\in N$.
		\item All  agents in $N_2(r_1,r_2)$ receive the same assignment $\chi(r_1,r_2)\in\{0,1\}$. In other words, $\hat{x}(i;h_2)=\chi(r_1(i),r_2(i))$ for all $i\in N$.
	\end{itemize}
	
	Note that condition (\ref{condi:aggregate}) requires
	\begin{align*}
		\sum_{r_1,r_2}\nu_2(r_1,r_2)\chi(r_1,r_2)\leq \lambda (S).
	\end{align*}
	
	\textit{A prima facie}, $s_2(r_1)$ and $\chi(r_1,r_2)$ can depend on the the entire history profile $(h_1,h_2)$. Anonymity only ensures that they are constant within each $N_t(h)$. For instance, they can depend on the	identity (i.e., the label or indexes) of the agents submitting those reports.\footnote{Consider the following testing recommendations for the second period: $s_1(r_1)=\mathbb{1}[sup_{i}N_t(r_1)\geq 0.4]$,  $s'_1(r_1)=\mathbb{1}[N_1(r_1)\subseteq[0,1/2]]$, and $s''(r_1)=\mathbb{1}[N_1(r_1)\text{ is convex}]$. The three recommendations satisfy anonymity but depend on the labels given to the agents.} We now proceed to rule-out that case. Formally, for any $\phi\in\Phi$ define
	\[N^\phi_t(h):=\{i\in N: h_t(\phi^{-1}(i))=h \}.\]
	Similarly, let $s^\phi_2(r_1)\in\{0,1\}$ be the testing recommendation in the second period given to agents in $N^\phi_1(r_1)$ and $\chi^\phi(r_1,r_2)\in\{0,1\}$ be the assignment given to agents in $N^\phi_2(r_1,r_2)$. By label-invariance,  we have that 
	\begin{itemize}
		\item $s_2(r_1)=s^\phi_2(r_1)$ for all $r_1$, and 
		\item $\chi(r_1,r_2)=\chi^\phi(r_1,r_2)$ for all $(r_1,r_2)$.
	\end{itemize}
	
	Moreover, consider any two \textit{collective} history profiles $(h_1,h_2)$ and $(h'_1,h'_2)$ inducing 	the same $\nu$.  Let $N'_t(h):=\{i\in N:h'_t(i)=h\}$. Then, there exists a $\phi\in\Phi$ such that $N'_t(h)=N_t^{\phi}(h)$.\footnote{Since $\mathcal{H}_t$ is finite, we can match each $N'_t(h)$ to an interval $\mathcal{I}'_t(h)$ and each $N_t(h)$  to an interval $\mathcal{I}_t(h)$. Define $\phi$ such those intervals are matched up to $\lambda$-null sets.} Thus, by the argument in the previous paragraph, label-invariance imply that $s_1$ and $\chi$ are the same for both  \textit{collective} history profiles $(h_1,h_2)$ and $(h'_1,h'_2)$. 
	
	Finally, we need to rule-out cases in which the behavior of a zero $\lambda$-measure subset of agents influences the recommendations and assignments of the mechanism.\footnote{For instance, consider the testing recommendation $s(r_1)=\mathbb{1}[N_1(r_1)\cap (0.2,0.4)\text{ is empty}]$. Note that this recommendation satisfy anonymity. Moreover, label-invariance does not rule it out either (since for any $\phi\in\Phi$ there will always be relabeled agents in the interval $(0.2,0.4)$).} Consider any two \textit{collective} histories $h_t$ and $h'_t$ such that for $\lambda$-.a.e. $i\in N$
	\[\lambda(\{i\in N:h_t(i)\neq h'_t(i)\})=0.\]
	In order to avoid ambiguities with the notation, let $N'_t(h):=\{i\in N:h'_t(i)=h\}$ and $\nu'_t$ be the corresponding empirical distribution of \textit{private} histories under $h'_t$. Since $N_t(h)$ and $N'_t(h)$ only differ on a set of $\lambda$-measure zero, it follows that $\nu_t=\nu'_t$. Furthermore, let 
	$s'_2(r_1)$ be the testing recommendation in the second period given to agents in $N'_1(r_1)$ and $\chi'(r_1,r_2)$ be the assignment given to agents in $N'(r_1,r_2)$.	By null-insensitivity, we have that 
	\begin{itemize}
		\item $s_2(r_1)=s'_2(r_1)$ for all $r_1$, and 
		\item$\chi(r_1,r_2)=\chi'(r_1,r_2)$  for all $(r_1,r_2)$.
	\end{itemize}

\end{proof}

Before stating that the principal's problem in many-agent setting can be reduced to the one-agent case, we introduce a final definition.

\begin{definition}[Individual mechanisms]
	A mechanism is said to be individual if for all $i\in N$: $\sigma_1(i)=\sigma^{ind}_1$, $\sigma_2(i;h_1)=\sigma^{ind}_2(r_1(i))$, and $\hat{x}(i;h_2)=\hat{x}^{ind}(r_1(i),r_2(i))$.
\end{definition}

\begin{proposition}
	Consider the case in which \autoref{assump:admissibility}, \autoref{assump:symmetric}, and \autoref{assump:exactaggregation} hold. Then, the outcomes generated by an admissible mechanism $m=(\sigma,\hat{x})$ can be replicated by 
	a collection $\{m(i)\}_{i\in N}$ of identical individual mechanisms $m(i)=(\sigma^{ind},\hat{x}^{ind})$ up to a subset of agents with $\lambda$-measure zero.
\end{proposition}

\begin{proof}
	By \autoref{lemma:equiv},  an admissible mechanism can condition the testing recommendations and assignments to a given agent $i$ on: $i$'s own reports, and the empirical distribution of \textit{private} histories $\nu$.  By \autoref{assump:exactaggregation}, $\nu=(\nu_1,\nu_2)$ is almost surely constant. Moreover, since the states are independently distributed, any agent (having observed his own test results) cannot learn anything about the other agents' test results or update his own conjecture about the strategies employed by other agents. Thus, $i$'s actions can only affect which group $N_t(h)$ he will belong to.  Note that the testing recommendation and assignments for one group can  depend on the testing recommendation and the assignments, respectively, given to the other groups. However, fixing a mechanism and the strategies of his peers, each agent can perfectly anticipate the consequences of each report he submits. 
	
	We construct an \textit{individual} mechanism that replicates the outcomes (in terms of testing decisions, reports, and assignments) of a given admissible mechanism as follows. Using \autoref{lemma:equiv}, let $\sigma_1^{ind}=\sigma_1$, $\sigma^{ind}_2(r_1(i))=s_2(r_1)$, and $\hat{x}^{ind}(r_1(i),r_2(i))=\chi(r_1,r_2)$. By \autoref{assump:symmetric}, we can use the exact same construction for almost all agents. 
	
	Note that the consequences of each agent's decision is exactly the same under $m$ and under $\{m(i)\}_{i\in N}$ where $m(i)=(\sigma^{ind,\hat{x}^{ind}})$. Moreover, the expected payoff from a deviation only depends on $\nu$ and his own history.  Hence, the agents' choices and outcomes are the same under both mechanisms.  
\end{proof}

	\pagebreak

\section{Additional results in the one-agent model}
\label{appen:additional}

\subsection{Optimal mechanisms}

\begin{proposition}
	\label{prop:implementability}
	If $\kappa,\gamma\in(0, 1-\rho]$, then the optimal mechanism coincides with the baseline.
\end{proposition}

\begin{proof}[Proof of \autoref{prop:implementability}]
	When $\kappa\leq 1-\rho$  the baseline recommends testing in the second period. After a report $r_1$ in the first period, $A$ is willing to follow such recommendation if and only if
	\[
	-c+\mu_2 \hat{x}(r_1,1)+(1-\mu_2)\hat{x}(r_1,0)\geq \pi [\mu_2 \hat{x}(r_1,1)+(1-\mu_2)\hat{x}(r_1,0)]+(1-\pi)\hat{x}(r_1,\O),
	\]
	equivalently,
	\[
	\gamma\leq \mu_2 \hat{x}(r_1,1)+(1-\mu_2)\hat{x}(r_1,0)-\hat{x}(r_1,\O),
	\]
	Note that setting $\hat{x}(r_1,\O)=0$ can only relax this constraint, and if $A$ follows the recommendation then $r_2=\O$ occurs with probability zero. Thus, this adjustment does not affect $P$'s expected payoff. Now, the efficiency of assignments is equivalent to setting $\hat{x}(r_1,r_2)=1$ if and only if $r_2=1$. Thus, $A$ is indifferent between reporting $r_2=0$ and $r_2=\O$ and strictly prefers to report $r_2=1$ rather than $r_2=\O$. With this in mind, the condition for $A$ to test in $t=2$ becomes $\gamma\leq \mu_2$. Since $\mu_2\geq 1-\rho$, it is clear that $A$ is willing to test in $t=2$ regardless of the history when $\gamma\leq 1-\rho$.
	Finally, note that neither the assignment nor the test recommendation in the following period depend on $r_1$, thus it is not optimal for the agent to test in $t=1$ whenever $\gamma> 0$.
\end{proof}

\begin{proposition}[Optimality of never testing]
	\label{prop:nevertesting}
	Assume $\kappa>\min\{\mu_2(\O),1-\mu_2(\O)\}$ and $\gamma\geq \rho$.
	Then, never testing is optimal and $\hat{x}(r_1,r_2)=\mathbb{1}[r_2=1]$.
\end{proposition}

\begin{proof}[Proof of \autoref{prop:nevertesting}]
	when $\kappa>\min\{\mu_2(\O),1-\mu_2(\O)\}$  the baseline mechanism recommends to never test. Clearly, the efficient assignments are not $IC$ when coupled with this testing policy. In particular, $A$ would strictly prefer to report $r_t=\O$. In order to satisfy $[FD,r_1,\omega_2=0]$  we must have $\hat{x}(r_1,0)\geq \hat{x}(r_1,\O)$ for all $r_1$.

Clearly, $\hat{x}(r_1,1)=0<\hat{x}(r_1,0)=1$ is never optimal. The following cases remain:

\begin{itemize}
	\item $\hat{x}(r_1,1)=1$ and $\hat{x}(r_1,0)= \hat{x}(r_1,\O)=0$ for all $r_1$.
	
	Then, $[OB,\sigma_2(r_1)=0]$ becomes 
	\begin{align*}
		\gamma\geq \mu_2[\hat{x}(r_1,1)-\hat{x}(r_1,0)]=\mu_2.
	\end{align*}
	
	Then $\gamma\geq \rho$ this is always satisfied, since $\mu_2\leq \rho$. Moreover, $v(0)=v(\O,1-\rho)=\pi(1-\rho)$. Thus, $[FD,\omega_1=0]$ is satisfied.
	
	In this case, 
	\[
	W=\pi+(1-\pi)[(1-\mu_0)\rho+\mu_0(1-\rho)]=\pi+(1-\pi)(1-\mu_2(\O))
	\]
	
	\item $\hat{x}(r_1,1)=\hat{x}(r_1,0)= \hat{x}(r_1,\O)=0$ for all $r_1$.
	
	This policy is $IC$ since the agent is indifferent between any report and testing is never optimal for him.
	In this case, 
	\[
	W=1-\mu_2(\O).
	\]
	This assignment policy is not optimal since $\mu_2(\O)<1$.

	\item $\hat{x}(r_1,1)=\hat{x}(r_1,0)= \hat{x}(r_1,\O)=1$ for all $r_1$.
	
	As in the previous case, this policy is $IC$ since the agent is indifferent between any report and testing is never optimal for him.
	In this case, 
	\[
	W=\mu_2(\O)
	\]
	Likewise, this assignment policy is not optimal since $\mu_2(\O)<1$.
\end{itemize}
\end{proof}

\subsection{Closing $P$'s lack of control channel}
\begin{lemma}[Implementability of the baseline when $\mathbf{\pi=0}$]
	\label{lem:imppi0}
	Assume $\kappa\in(1-\rho,\min\{\mu_2(\O),1-\mu_2(\O)\}]$, $\gamma\leq \mu_2(\O)$, and $\pi=0$. Then, the mechanism with $\hat{x}(r_1,r_2)=\mathbb{1}[r_1=\O]\mathbb{1}[r_2=1]$, $\sigma_1=0$, and $\sigma_2(r_1)=\mathbb{1}[r_1=\O]$ generates the same outcomes as the baseline mechanism, and therefore it is optimal.
\end{lemma}
\begin{proof}
	 Clearly, $A$ follows the testing policy since any deviation is detected and punished with his least preferred assignment. Likewise, $A$ is, at worst, indifferent between disclosing the observed result in $t=2$ and $r_2=\O$.	 

\end{proof}

\subsection{Gaming the Baseline Under $\gamma \in (0,1-\rho]$ and Intermediate $\kappa$}
\label{subsec:lowgamma}
 Consider the following \textit{deviation} from the baseline recommendation.
 \begin{definition}[Deviation Strategy]
 		\label{def:deviation2}
 	$A$ tests in $t=1$, submit its result if and only if $\omega_1=1$. In that case, $A$ does not test in $t=2$. If a free test is obtained in $t=2$ report its result if and only if $\omega_2=1$. $A$ re-tests in $t=2$ if $\omega_1=0$, and submits any test result. 
 \end{definition}
 
 The following proposition establishes that when $A$’s effective testing costs fall below a given threshold, he strictly prefers to deviate from the baseline testing recommendations.
 \begin{assumption}[Low testing costs for $A$]
 	\label{assump:lowGamma}
 	\[\gamma\leq 1-\rho.\]
 \end{assumption}
\begin{proposition}
	\label{proposition:profdev}
	Consider the case in which \autoref{assump:interKappa} and \autoref{assump:lowGamma} hold. The  deviation in \autoref{def:deviation2} is strictly profitable for $A$ if and only if
	\begin{align*}
		\gamma< \bar{\gamma}':=\frac{\mu_0(1-\rho)}{1-\mu_0(1-\pi)}.
	\end{align*}
\end{proposition}

\begin{proof}
	If $A$ follows $\sigma^*_1=0$ he can expect a free test is $t=1$ with probability $\pi$. 
	
	If the free result is $\omega_1=1$ the agent reports $r_1=1$ and $r_1=\O$ otherwise. If $r_1=1$, the assignment is 1, regardless of the report in $t=2$ (since $\omega_2=0$ would be concealed). 
	
	If the free test result is $\omega_1=0$ then $A$ optimally sends $r_1=\O$ which leads to mandatory testing in $t=2$. Since $\gamma\leq 1-\rho$, $A$ is willing to follow $\sigma_2(\O)=1$ regardless of weather information was concealed. 
	
	If $r_1=\O$, $A$ can expected a payoff of $-c+\mu_2$, since $\hat{x}(\O,r_2)=1$ if and only if $r_2=1$ (remember, we are considering that a detected deviation, such as sending $r_2=\O$ when $\sigma_2(\O)=1$, are punished with the worst allocation). Thus, $A$ can expect a payoff of $-c(1-\pi\mu_0)+\pi[1-\rho(1-\mu_0)]+(1-\pi)\mu_2(\O)$ if he follows $\sigma^*_1=0$.
	
	If $A$ deviates to the strategy prescribed above, he incurs in testing costs of $c$ in $t=1$ and in $t=2$ if and only if $\omega_1=0$, which occurs with probability $1-\mu_0$. If $r_1=1$, $A$ is guaranteed an allocation of 1. If $r_1=\O$, equivalently $\omega_1=0$, then $A$ expects an allocation of $1-\rho$. Thus, the expected payoff for $A$ from this strategy is $-c(2-\mu_0)+1-\rho(1-\mu_0)$.
	
	Thus, $A$ (strictly) prefers to deviate from $\sigma^*_1=0$ if and only if $-c(2-\mu_0)+1-\rho(1-\mu_0)>-c(1-\pi\mu_0)+\pi[1-\rho(1-\mu_0)]+(1-\pi)\mu_2(\O)$. Re-arranging this expression and using the fact that $1-\mu_2(\O)=\mu_0(1-\rho)+(1-\mu_0)\rho$ we get the desired result.
\end{proof}

Note that $\bar{\gamma}'\leq1-\rho$ if and only if $\mu_0\leq1/(2-\pi)\in(1/2,1)$. Hence, when $\mu_0> 1/(2-\pi)\in(1/2,1)$, the deviation strictly benefits $A$ for all $\gamma\leq 1-\rho$.

\end{document}